\documentclass[letterpaper, 11pt]{article}
\usepackage{latexsym}
\usepackage{amssymb}
\usepackage{times}
\usepackage[in]{fullpage}
\usepackage{amsmath,amsfonts,amsthm}
\usepackage{enumerate}

\newcommand{\ble}{\begin{lemma}}
\newcommand{\ele}{\end{lemma}}

\newtheorem{lemma}{Lemma}[section]

\newtheorem{theorem}[lemma]{Theorem}

\newtheorem{definition}[lemma]{Definition}
\newtheorem{corollary}[lemma]{Corollary}

\newtheorem{example}[lemma]{Example}

\newtheorem{fact}[lemma]{Fact}
\newtheorem{algorithm}[lemma]{Algorithm}

\newcommand{\beao}{\begin{eqnarray*}}
\newcommand{\eeao}{\end{eqnarray*}\noindent}
\newcommand{\beam}{\begin{eqnarray}}
\newcommand{\eeam}{\end{eqnarray}\noindent}

\newcommand{\one}{{\bf 1}}

\begin{document}
\def\DoubleSpace{\baselineskip=24pt}
\DoubleSpace


\title{ Measuring Independence of Datasets}
\author{Vladimir Braverman \ \ \ \ \ \ \ Rafail Ostrovsky\\
  \textrm{University of California Los Angeles}\\
  \texttt{\{vova, rafail\}@cs.ucla.edu} \\ \ \\  }

\maketitle

\thispagestyle{empty}
\begin{abstract}
\noindent
Measuring independence between two or more random variables is a
fundamental problem that touches many areas of computer science. The problems
of  efficiently testing pairwise, or $k$-wise, independence were recently
considered by Alon, Andoni, Kaufman, Matulef, Rubinfeld and Xie (STOC 07);
Alon, Goldreich and Mansour (IPL 03);
Batu, Fortnow, Fischer, Kumar, Rubinfeld and White (FOCS 01); and
Batu, Kumar and Rubinfeld (STOC 04).
They addressed the problem of minimizing the number of samples needed to
obtain sufficient approximation, when the joint distribution is accessible
through a sampling procedure.

\bigskip
\noindent
A data stream model represents another setting where approximating pairwise,
or $k$-wise, independence with sublinear memory is of considerable importance. Unlike the work in the aforementioned papers, in the streaming model
the joint distribution is given by a stream of $k$-tuples, with the goal of testing correlations among the components measured over the entire stream.
In the streaming model, Indyk and McGregor (SODA 08) recently gave exciting new
results for measuring pairwise independence.

\bigskip
\noindent
\emph{Statistical distance} is one of the most fundamental metrics for
measuring the similarity of two distributions, and it has been a metric of
choice in many papers that discuss distribution closeness (see, for example,
Rubinfeld and Servedio (STOC 05); Sahai and Vadhan (JACM 03); and the above
papers). The Indyk and McGregor methods provide $\log{n}$-approximation under
statistical
distance between the joint and product distributions in the streaming model.  (In contrast, for the $L_2$ metric, Indyk and
McGregor give an $(1\pm \epsilon)$-approximation for the same problem, but for probability
distributions, statistical distance is a significantly more powerful metric
then the $L_2$ metric). For the $L_1$ metric, in addition to $\log n$ approximation,
Indyk and McGregor give an $\epsilon$-approximation that requires linear
memory, and also give a method that requires two passes to solve a promise
problem for a restricted range of parameters. Indyk and McGregor leave, as
their main open question, the problem of  improving their $\log
n$-approximation for the statistical distance metric.

\bigskip
\noindent
In this paper we solve the  main open problem  posed by of Indyk and McGregor
for the statistical distance for pairwise independence and extend this result to any constant $k$. In particular, we present an algorithm that computes
an $(\epsilon, \delta)$-approximation
of the statistical distance between the joint and product distributions defined by a stream of $k$-tuples.
Our algorithm requires $O(\left({1\over \epsilon}\log({nm\over \delta})\right)^{(30+k)^k})$ memory and a single pass over the data stream.
\end{abstract}

\newpage
\setcounter{page}{1}
\sloppy
\section{Introduction}
Finding correlations between columns of a table is a fundamental problem in databases.
Virtually all commercial databases construct query plans for queries that employ cross-dimensional predicates.
The basic step is estimating ``selectivity'' (i.e., the number of rows that satisfy the predicate conditions)
of the complex predicate. Without
any prior knowledge, the typical solution is to compute selectivity of each column separately
and use the multiplication as an estimate. Thus, optimizers make a ``statistical independence
\linebreak
assumption''
which sometimes may not hold. Incorrect estimations may lead to suboptimal query plans and decrease
performance significantly. Identifying correlations between database columns by measuring a level of independence between columns has a long history
in the database research community.
To illustrate this point, we cite as an example,
Poosala and Ioannidis \cite{poosala}:

\begin{quote}
{\em ``For a query involving two or more attributes of the same
relation, its result size depends on the joint data distribution
of those attributes; i.e., the frequencies of all combinations
of attribute values in the database. Due to the
multi-dimensional nature of these distributions and the large
number of such attribute value combinations, direct approximation
of joint distributions can be rather complex and expensive.
In practice, most commercial DBMSs adopt the
attribute value independence assumption. Under this
assumption, the data distributions of individual attributes in
a relation are independent of each other and the joint data
distribution can be derived from the individual distributions
(which are approximated by one-dimensional histograms).
Unfortunately, real-life data rarely satisfies the attribute
value independence assumption. For instance, functional
dependencies represent the exact opposite of the assumption.
Moreover, there are intermediate situations as well.
For example, it is natural for the salary attribute of
the Employee relation to be 'strongly' dependent on the
age attribute (i.e., higher/lower salaries mostly going to
older/younger people). Making the attribute value independence
assumption in these cases may result in very inaccurate
approximations of joint data distributions and therefore
inaccurate query result size estimations with devastating effects
on a DBMS's performance.''}
\end{quote}

For data warehouses, it is important to find correlated columns for correct schema construction, as Kimball and Caserta note in \cite{kimbal}:
\begin{quote}
{\em ``Perfectly correlated attributes, such as the levels of a hierarchy, as well as attributes with a
reasonable statistical correlation, should be part of the same dimension.''}
\end{quote}

In practice, typical solutions for finding correlations between columns are either histograms (see e.g., \cite{poosala})
or sampling (see e.g., \cite{cords}). These methods have their natural disadvantages, i.e.,
they do not tolerate deletions and may require several passes over the data.
When it comes to very large data volumes, it is critical to maintain sublinear in terms of memory solutions
that do not require additional passes over the data and can tolerate incremental updates of the data,
e.g., deletions.

For these purposes, a theoretical \emph{data stream model} can be useful.
For data warehouses, the ``loading'' phase of the ETL process (see e.g., Kimball and Caserta \cite{kimbal})
can be seen as a data stream.
When reading a database table, the process can be considered as a stream of data tuples.
Thus, the data stream model represents another setting where approximating pairwise
or $k$-wise independence with sublinear memory is of considerable importance.

\subsection{Precise Definition of the Problem}
The natural way to model database tables in a streaming model is by considering a stream of tuples.
In this paper we consider a stream of $k$-tuples $(i_1,\dots,i_k)$ where $i_l \in [n]$.
(For simplicity, we assume that elements of all columns are drawn from the same domain,
even though our approach trivially extends to a general case of different domains.)
As pointed out in \cite{poosala, cords, mi}, the natural way to define a
joint distribution of two (or more) columns is given by the frequencies of all combinations of coordinates.
Similarly, the distribution of each column is defined by the corresponding set of frequencies;
the definition of a product distribution follows. Let us define these notions precisely.\footnote{Here and thenceforth, we use lowercase Latin characters for indexes.
We use an italic font for integers and a boldface font for
 multidimensional indexes, e.g., $i\in [n]$ and $\mathbf{i} \in [n]^k$. For a multidimensional index, we use subscript
 to indicate its coordinate, e.g., $\mathbf{i}_1$ indicates the first coordinate of $\mathbf{i}$.
}
\begin{definition}\label{def: independence problem a}
Let $D$ be a stream of elements $p_1,\dots, p_m$, where each stream element is a $k$-tuple ${\mathbf{i}} = ({\mathbf{i}}_1,\dots,{\mathbf{i}}_k)$, where ${\mathbf{i}}_l \in [n]$.
A \emph{frequency} of a tuple ${\mathbf{i}} \in [n]^k$ is defined as the number of times it appears in $D$:
$
f_{{\mathbf{i}}} = |\{j: p_j = {\mathbf{i}}\}|.
$
For $l\in [k]$, a \emph{$l$-th margin frequency} of $t \in [n]$ is the number of times $t$ appears as a $l$-th coordinate:
$
f_{l}(t) = \sum_{{\mathbf{i}}\in [n]^k, {\mathbf{i}}_l = t}^n f_{{\mathbf{i}}}.
$
A \emph{joint distribution} is defined by a vector of probabilities
$
P_{joint}({\mathbf{i}}) = {f_{{\mathbf{i}}} \over m}, {\mathbf{i}}\in [n]^k.
$
Here $m$ is the size of stream $D$. A \emph{$l$-th margin distribution} is defined by a vector of probabilities
$
P_{l}(t) = {f_{l}(t) \over m}, \ t\in [n].
$
A \emph{product distribution} is defined as:
$
P_{product}({\mathbf{i}}) = \prod_{l=1}^k P_{l}({\mathbf{i}}_l), {\mathbf{i}}\in [n]^k.
$
\end{definition}

\emph{Statistical distance} is one of the most fundamental metrics for
measuring the similarity of two distributions, and it has been a metric of
choice in many papers that discuss distribution closeness (see e.g., \cite{k-wise_independenc, ind3, batu_independence, ind2, mi, sahai, ind4}).
Given two distributions over a discrete domain,
the statistical distance is half of $L_1$ distance between the probability vectors.
\begin{definition}
Consider two distributions over a finite domain $\Omega$ given by two random variables $V, U$.
\emph{Statistical distance} $\Delta(V,U)$ is defined as:
$$
\Delta(V, U) = {1\over 2}\sum_{x\in \Omega} |P(V= x) - P(U=x)| = \max_{B\subseteq \Omega} |P(V\in B) - P(U\in B)|.
$$
\end{definition}

In particular, one of the most common methods of measuring independence is computing statistical distance between product
and joint distributions (see e.g., \cite{batu_independence, mi}). This is precisely the way we define our problem:
\begin{definition}
An \emph{Independence Problem} is the following:
Given stream $D$ of $k$-tuples, approximate, with one pass over $D$, with small memory and high precision the statistical distance between joint and product distribution
$
\Delta(P_{joint}, P_{product}).
$
\end{definition}

In the streaming model, Indyk and McGregor \cite{mi} recently gave exciting new
results for measuring pairwise independence, i.e., for $k=2$. To measure the independence, they consider two metrics: $L_2$ and $L_1$.
Recall that the $L_2$ distance between two probability distributions is a $L_2$ distance of their probability vectors. In particular,
the independence problem under the $L_2$ metric is defined as
$
\|P_{joint}-P_{product}\|_2.
$

For the $L_2$ metric and $k=2$, Indyk and McGregor give an $(1\pm \epsilon)$-approximation using polylogarithmic space.
However, it is well known that for probability distributions,
statistical distance is a significantly more powerful metric then the $L_2$ metric.
For instance, consider two distributions on $[2n]$, where the first distribution is uniform on $\{1,\dots, n\}$
and the second is uniform on $\{n+1,\dots, 2n\}$. In this case the statistical distance is $1$ but the $L_2$ distance is $\sqrt{2/n}\rightarrow 0$.
For example, Batu, Fortnow, Rubinfeld, Smith and White \cite{ind5} say:
\begin{quote}
{\em ``However,
the $L_2$-distance does not in general give a good measure of
the closeness of two distributions. For example, two distributions
can have disjoint support and still have small $L_2$-
distance.''}
\end{quote}

For the statistical distance metric and $k=2$, the Indyk and McGregor methods provide $\log{n}$-approximation
with polylogarithmic memory. In addition to $\log{n}$-approximation,
Indyk and McGregor give an $(1\pm \epsilon)$-approximation that requires $\Omega(n)$
memory, and also give a method that requires two passes to solve a promise
problem for a restricted range of parameters. Indyk and McGregor leave, as
their main open question, the problem of  improving their $\log
n$-approximation for the statistical distance metric.

In this paper we solve the  main open problem  posed by of Indyk and McGregor
for the statistical distance for pairwise independence and extend this result to any constant $k$. In particular, we present an algorithm that computes
an $(\epsilon, \delta)$-approximation
of the statistical distance between the joint and product distributions defined by a stream of $k$-tuples.
Our algorithm requires $O(\left({1\over \epsilon}\log({nm\over \delta})\right)^{(30+k)^k})$ memory and a single pass over the data stream. Theorem \ref{tm: main theorem} formally describes our main result. We did not try to optimize the constants in our memory bounds.

\subsection{Implicit Tensors}

It is convenient to present an alternative, equivalent formulation of the {\em independence problem} as well. We can consider the problem of approximating
the sum of absolute values of a \emph{tensor} $M_{Ind}$.

\begin{definition}
An $s$-dimensional \emph{tensor} $M$ is a $s$-dimensional array with indexes in the range $[n]$; that is, $M$ has an entry for each ${\mathbf{i}}\in [n]^s$. We denote by $m_{\mathbf{i}}$ the ${\mathbf{i}}$-th entry of $M$ for each ${\mathbf{i}}\in [n]^s$.
\end{definition}

\begin{definition}
Let $M$ be a $s$-dimensional tensor  with entries $m_{{\mathbf{i}}}, {\mathbf{i}}\in [n]^s$.
An $L_1$-norm of $M$ is a  $|M| = \sum_{{\mathbf{i}} \in [n]^{s}} |m_{{\mathbf{i}}}|$.
\end{definition}

For example, a $1$-dimensional tensor is an $n$-dimensional vector, a $2$-dimensional tensor is an $n\times n$-matrix and so forth.

Many streaming problems address \emph{explicitly} defined vectors (or matrices) where entries are equal to   frequencies of corresponding stream elements.
The Independence problem diverges from this setting; e.g., for pairwise independence, a pair $(i,j)$ affects all entries in $i$-th row and $j$-th column of the product probability matrix.
To reflect this important difference we consider the case where the entries of a tensor are defined \emph{implicitly} by a data stream.

\begin{definition}\label{def:wider ddcccddddddd}
Let $\mathcal{D}$ be a collection of data streams of size $m$ of elements from domain $\Omega$.
Let $\mathcal{F}: \mathcal{D} \times [n]^s \mapsto R$ be a fixed function. We say that $s$-dimensional tensor $M$ with entries $m_{\mathbf{i}} = \mathcal{F}(D, {\mathbf{i}}), {\mathbf{i}}\in [n]^s$ is \emph{implicity defined} by $\mathcal{F}$, given $D$. We denote an implicitly defined tensor as $\mathcal{F}(D)$.
\end{definition}

\begin{definition}\label{def:wider ddcccddddddffffffd}
Let $\mathcal{D}$ be a collection of data streams of size $m$ of $k$-tuples from domain $[n]^k$.
A $k$-wise \emph{Independence Function} $\mathcal{F}_{Ind}: \mathcal{D}\times [n]^k \mapsto R$ is a function defined as $\mathcal{F}_{Ind}(D, {\mathbf{i}}) = m^kf_{{\mathbf{i}}} - \prod_{l=1}^k f_{l}({\mathbf{i}}_l)$ for $\mathbf{i}\in [n]^k$. Here $f_{\mathbf{i}}$ is given by Definition \ref{def: independence problem a}.
\emph{Statistical distance tensor} $M_{Ind}$ is a $k$-dimensional tensor implicitly defined by $\mathcal{F}_{Ind}$, i.e., $M_{Ind} = \mathcal{F}_{Ind}(D)$.
\end{definition}

The main objective of our paper is approximating $|M_{Ind}|$. In particular, this implies solving  the Independence problem since $\Delta(P_{joint}, P_{product}) = {1\over m^k}|M_{Ind}|$, and since $m = |D|$ can be computed precisely.
We thus freely interchange the notions of the independence problem and computing $|M_{Ind}|$.
In fact, our approach is applicable to any function $\mathcal{F}$ for which
conditions of our main theorems are true.

\subsection{Why Existing Methods for Estimating $L_1$ Do Not Work}

Alon, Matias and Szegedy \cite{ams} initiated the study of computing norms of vectors defined
by a data stream. In their setting vector entries are defined by frequencies of the corresponding elements in the stream.
Their influential paper was followed by a sequence of exciting results including, among many others, works by
Bhuvanagiri, Ganguly, Kesh and Saha \cite{frequency_moments};
Charikar, Chen and Farach-Colton \cite{frequent}; Cormode and Muthukrishnan \cite{counmin, deltoids};
Feigenbaum, Kannan, Strauss and Viswanathan \cite{feigenbaum};
Ganguly and Cormode \cite{cormode ganguly};
Indyk \cite{stable};
Indyk and Woodruff \cite{frquency_moments1};
and Li \cite{li1, li2}.

There is an important difference between settings of \cite{ams} and the Independence problem.
Indeed, while the entries of the independence tensor are defined by frequencies of tuples,
there is no linear dependence. As a result, the aforementioned algorithms are not directly applicable to the Independence problem.

To illustrate this point, consider the celebrated method of stable distributions by Indyk \cite{stable}.
For $L_1$ norm, Indyk observed that a polylogarithmic (in terms of $n$ and $m$) number of sketches of the form $\sum C_{i} v_i$ gives an $(1\pm \epsilon)$-approximation
of $|V|$, when $C_i$ are independent random variables with Cauchy distribution.
Let us discuss the applicability of this method to the problem of pairwise independence.
A sketch $\sum_{{\mathbf{i}}} C_{{\mathbf{i}}}m_{{\mathbf{i}}}$, ${\mathbf{i}}\in [n]^2$, would solve this problem;
unfortunately, it is not clear how to construct a sketch in this form.
In particular, the probability matrix of the product distribution is given implicitly as two vectors of margin sketches.
It is not hard to construct sketches for margin distributions; however, it is not at all clear how to obtain a sketch for product distribution without using a multiplication of margin sketches.
On the other hand, if we do use a multiplication of margin sketches (this is the approach of Indyk and McGregor),
the random variable that is associated with the tensor's elements is a product of independent Cauchy variables.
Therefore, random variables for distinct entries are {\em not independent}, and thus typical arguments used for stable distribution methods do not work anymore.
In fact, the main focus of the Indyk and McGregor analysis is to overcome this problem:
\begin{quote}
{\em ``Perhaps ironically, the biggest technical challenges that arise
relate to ensuring that different components of our estimates are sufficiently independent.''}
\end{quote}
For pairwise independence, Indyk and McGregor use the product of two Cauchy variables, where one of them is ``truncated.''
Using elegant observations, they show that such a sketch allows achieving $\log{n}$-approximation of the statistical distance.
Unfortunately, it is not clear how the method of a Cauchy product can be improved at all,
since the $\log{n}$ factor is a necessary component of their seemingly tight analysis.

\subsection{A Description of Our Approach}
As we discuss below, solving the Independence problem requires developing multiple new tools and using them jointly  with known methods.

\bigskip
\noindent

\bigskip
\noindent
\underline{\emph{Dimension Reduction for Implicit Tensors}}.
Our solution can be logically divided into three steps which are explained, informally, below.

First, we prove that given a $polylog$-approximation algorithm for $k$-dimensional tensors and an $\epsilon$-approximation algorithm for a special type of $(k-1)$-dimensional tensors, it is possible to derive an $\epsilon$-approximation algorithm on $k$-dimensional tensors, where the resulting algorithm increases memory bound by a factor $O(({1\over\epsilon}\log{nm\over \delta})^{O(1)})$.
Thus, we can trade dimensionality and precision for memory.
To illustrate this step, consider pairwise independence.
There exist an $\epsilon$-approximation algorithm on vectors \cite{stable}
and a $\log{n}$-approximation algorithm on matrices \cite{mi}.
We show that these algorithms can be used to obtain an $\epsilon$-approximation algorithm on matrices.
This informal idea is stated precisely as Dimension Reduction Theorem \ref{tm: fff reduction}.
This theorem is the main technical contribution of our paper; the majority of
the paper is devoted to establishing its validity.

Second, given a $polylog$-approximation algorithm for $k$-dimensional tensors and
an $\epsilon$-approximation algorithm on vectors, we can derive an $\epsilon$-approximation
algorithm on $k$-dimensional tensors  by applying the Dimension Reduction Theorem recursively $k$-times. 
The memory will be increased by a factor roughly $O(({1\over\epsilon}\log{nm\over \delta})^{(30+k)^k})$
which is $O(({1\over\epsilon}\log{nm\over \delta})^{O(1)})$ for constant $k$.
This informal idea is stated precisely as Theorem \ref{tm: approxima theorem}.

Third, we show that the conditions for Theorem \ref{tm: approxima theorem} hold for the Independence problem.
These results are stated in Lemmas \ref{lm: polylog for indep tensors}
and \ref{lm: epsilon for indep tensors}, and in fact are a generalization of
results from \cite{stable, mi}. Section \ref{sec:indyk genralizing} is devoted to the proof of these lemmas.

The rest of our discussion is devoted to a description of the main ideas behind the Dimension Reduction Theorem.

\bigskip
\noindent
\underline{\emph{Hyperplanes and Absolute Vectors}}.
Consider a matrix $M$; a very natural idea to approximate $|M|$ is by approximating a $L_1$ norm of a vector with entries equal to $L_1$ norms of rows of $M$.
We generalize this idea to tensors by defining the following operators.

\begin{definition}\label{def: hyperplaned hypercubcccccccce}
For any $s,t\ge 0$, we denote by $(,)$ a mapping from $[n]^{s}\times [n]^{t}$ to $[n]^{s+t}$ obtained by concatenation of coordinates. For instance, $((1,2), 3)$ is a an element from $[n]^3$ with coordinates $1,2,3$ respectfully.
\end{definition}

\begin{definition}\label{def: hyperplaned tensor}
Let $M$ be a $s$-dimensional tensor with entries $m_{{\mathbf{j}}}, {\mathbf{j}}\in [n]^s$.
For any $l\in [n]$, $Hyperplane(M,l)$ is a $(s-1)$-dimensional tensor with entries $m_{(l,{\mathbf{i}})}$
for ${\mathbf{i}} \in [n]^{s-1}$.
\end{definition}

\noindent
For example, when $k=2$, the $l$-th hyperplane of a matrix $M$ is its $l$-th row.

\begin{definition}
An $l$-th hyperplane is $\alpha$-\emph{significant} if  $|Hyperplane(M,l)| \ge \alpha|M|$.
\end{definition}

\noindent
For example, when $k=2$, the $l$-th row is $\alpha$-significant if the $L_1$-norm of the vector defined by the $l$-th row carries at least $\alpha$-fraction of $|M|$.

\begin{definition}\label{def:rollup tensfsssssor deded}
For  a $s$-dimensional tensor $M$, an $AbsoluteVector(M)$ is a vector of dimensionality $n$ with entries $|Hyperplane(M, l)|, l\in [n]$. In particular, $|AbsoluteVector(M)| = |M|$.
\end{definition}

\bigskip
\noindent
\underline{\emph{Projected Dimensions}}.
To prove Dimension Reduction Theorem \ref{tm: fff reduction} we need to map $s$-dimensional tensors to $(s-1)$-dimensional tensors with a small distortion of $L_1$. We come up with the following mapping.

\begin{definition}\label{def:rollup tensor deded}
Let $M$ be a $s$-dimensional tensor with entries $m_{\mathbf{l}}$, where ${\mathbf{l}} \in [n]^s$, and let $0\le t\le s$.
A \emph{Suffix-Sum} tensor $T_t(M)$ is a $(s-t)$-dimensional tensor with entries (
for each ${\mathbf{i}}\in [n]^{s-t}$):
$$m'_{{\mathbf{i}}} = \sum_{{\mathbf{j}}\in [n]^t} m_{({\mathbf{j}},{\mathbf{i}})}$$
Also, we define $T_0(M) = M$.
In other words, the ${\mathbf{i}}$-th entry of $T_t(M)$ is obtained by summing all elements of $M$ with the $(s-t)$-suffix equal to ${\mathbf{i}}$.
In particular, $T_s(M)$ is a scalar that is equal to $\sum_{{\mathbf{i}}\in [n]^s} m_{{\mathbf{i}}}$.
\end{definition}

For matrix $M$ with entries $m_{i,j}$, the Suffix-Sum operator $T_1(M)$ defines a vector $V$ with entries
$v_j = \sum_{i} m_{i,j}$. In other words, all entries of $M$ that belong to the same columns (i.e., have the same second coordinate, i.e., the same ``suffix'')
are ``summed-up'' to generate a single entry of $V$.
In some sense, the Suffix-Sum operator is orthogonal to the AbsoluteVector operator. In the latter case we sum up the absolute values that belong to the same hyperplane, i.e., have identical prefix;
in the former case we sum up all elements (and not their absolute values) that have an identical suffix.

Clearly $|T_1(M)| \le |M|$; however, it is possible in general that $|T_1(M)|\ll |M|$.
The key observation is that in some cases $|T_1(M)| \sim |M|$ and thus we can use an approximation of $|T_1(M)|$ to approximate $|M|$.
To illustrate this point, consider a matrix $M$ with entries $m_{i,j}$ that contains a very ``significant'' row $i$ (i.e., $\sum_{j} |m_{i,j}| \sim |M|$).
The key observation is that in this case
$|T_1(M)| \sim |M|$; thus, if there is a significant row, it can approximated using $|T_1(M)|$.
The same idea is easily generalized: if a $s$-dimensional tensor $M$ contains a $(1-\epsilon)$-significant hyperplane $Hyperplane(M, l)$, then $|T_1(M)|$ is an $2\epsilon$-approximation of $|Hyperplane(M, l)|$.
We prove this statement in Fact \ref{lm: tensor significant fff}.

Note that $T_1(M)$ is a $(s-1)$-dimensional tensor; if $M$ is a matrix, then $T_1(M)$ is a vector for which we can apply methods from \cite{stable}. Thus, approximating $|T_1(M)|$ is potentially an easier problem.

\bigskip
\noindent
\underline{\emph{Certifying Tournaments}}.
We have shown that $T_1(M)$ can be useful for approximating $|M|$.
However, when can we rely on the value of $|T_1(M)|$?
In particular, how can we distinguish between the cases when there is a heavy hyperplane (and thus $|T_1(M)|$ is a good approximation) and
the case when there is no heavy hyperplane (and thus $|T_1(M)|$ does not contain reliable information)?
 The second key observation is that it can be done using ``certifying tournaments.''
To illustrate this point, consider again the case $k=2$, where $M$ is a matrix.
Split $M$ into two random sub-matrices by sampling the rows w.p. $1/2$.
If there is a heavy row, then with probability close to $1$,
one sub-matrix will have a significantly larger norm then the other.
Recall that the method of \cite{mi} gives us a $\log{n}$-approximation.
Thus, for very heavy rows, the \emph{ratio} between approximations of norms obtained by the method from \cite{mi} will be large.
On the other hand, we show that if there are no heavy rows, then such behavior is quite unlikely to be observed many times.
Thus, there exists a way to distinguish between the first and the second cases for $(1-{\epsilon\over\log^2{n}})$-significant rows.\footnote{It is worth noting that the idea of ``split-and-compare'' is not new. Group testing \cite{deltoids} exploits
a similar approach. However, the methods from \cite{deltoids} require $\epsilon$-approximation of $L_1$; in contrast, we use certifying tournaments to improve the approximation.}

The method of certifying tournaments can be generalized to any $s\le k$ as follows.
Let $M$ be a $s$-dimensional tensor with entries $m_{\mathbf{i}}$ for ${\mathbf{i}}\in [n]^s$.
We ``split'' $M$ into two ``sampled'' $s$-dimensional tensors $M^0$ and $M^1$ by randomly sampling the first coordinate.
That is, $M^1$ has entries $m_{\mathbf{i}}H({\mathbf{i}}_1)$ and $M^0$ has entries $m_{\mathbf{i}}(\one-H({\mathbf{i}}_1))$,
where $H:[n]\mapsto \{0,1\}$ is pairwise independent and uniform.
If there exists a $\beta$-approximation algorithm for sampled tensors, and
there exists an $\epsilon$-approximation algorithm for Suffix-Sum, $|T_1(M^0)|$ and $|T_1(M^0)|$,
then we can approximate $L_1$ norm of significant hyperplanes.
Indeed, if there exists a significant hyperplane $M_l$ of $M$, then the ratio between $\beta$-approximations of $|M^0|$ and $|M^1|$
will be large. If this is the case, the approximation of $T(M^{H(l)})$ is also an $\epsilon$-approximation of
$|M_l|$.

To summarize, our main technical Theorem \ref{lm: the most importnat one tensor} proves that it is possible to
output a number $U$ such that $U$ is either an approximation of some hyperplane or $0$. Further, if there exists a
$(1-{\epsilon\over\beta^2})$-significant hyperplane, then with high probability, $U$ is its approximation.
We call such an algorithm an $\alpha$-ThresholdMax algorithm, for $\alpha = O({\epsilon\over\beta^2})$.

\bigskip
\noindent
\underline{\emph{Indirect Sampling}}.
Many streaming algorithms compute statistics on \emph{sampled} streams, which are random subsets of $D$ defined by some randomness $\mathcal{H}$.
In many cases, a sampled stream directly corresponds to a collection of sampled entries of a frequency vector.
In contrast, subsets of $D$ do not correspond directly to entries $M_{Ind}$.
Thus, our algorithms employ \emph{indirect} sampling, where randomness defines sampled entries of $M_{Ind}$ rather then the entries of a data stream $D$.
We define a Prefix-Zero operator.
\begin{definition}\label{def:wider dd}
Let $M$ be a $s$-dimensional tensor with entries $m_{\mathbf{i}}, {\mathbf{i}}\in[n]^s$ and let $H_1,\dots,H_t, t\le s$ be hash functions $H_j: [n]\mapsto \{0,1\}$. A \emph{Prefix-Zero} tensor  $W(M, H_1,\dots, H_t)$  a is a $s$-dimensional tensor with entries
$m_{{\mathbf{i}}} \prod_{l=1}^t H_l({\mathbf{i}}_l)$.
\end{definition}

Our algorithms work with tensors that are defined by compositions of $\mathcal{F}_{Ind}$, Prefix-Zero and Suffix-Sum.
We thus extend the definition of implicitly defined tensors.
\
\

\

\noindent
\textbf{Definition \ref{def:wider ddcccddddddd}.}\ \ \textbf{(Revised)}
\textit{Let $\mathcal{D}$ be a collection of data streams of size $m$ of elements
from domain $\Omega$ and let $\mathfrak{H}$ be a collection of hash functions from $[n]$ to $\{0,1\}$.
Let $\mathcal{F}: \mathcal{D} \times \mathfrak{H}^{t} \times [n]^s \mapsto R$ be a
fixed function, for some $0\le t \le s$. We say that a $s$-dimensional tensor $M$
with entries $m_{\mathbf{i}} = \mathcal{F}(D, \mathcal{H}, {\mathbf{i}}), {\mathbf{i}}\in [n]^s$
is \emph{implicity defined} by $\mathcal{F}$, given $D\in \mathcal{D}$ and $\mathcal{H} \in \mathfrak{H}^t$.
We denote an implicitly defined tensor as $\mathcal{F}(D, \mathcal{H})$.
}

\noindent
\begin{example}
Consider $k=2$. Then $\mathcal{F'}(D, H) = W(\mathcal{F}_{Ind}(D), H)$ defines a matrix
that represents a collection of rows sampled by a hash function $H:[n]\mapsto \{0,1\}$.
\end{example}

\bigskip
\noindent
\underline{\emph{Generalizing the Method of Indyk and Woodruff \cite{frquency_moments1} to Work on Implicit Vectors}}.
The ThresholdMax algorithm solves the problem that resembles the well-known problem of finding an element with maximal frequency,
see, e.g., \cite{frequent} and \cite{counmin}.
The celebrated method of Indyk and Woodruff \cite{frquency_moments1} uses maximal entries to estimate $L_p$ norms on vectors defined by frequencies. We apply the ideas of \cite{frquency_moments1} to approximate $|AbsoluteVector(M)| = |M|$.

Unfortunately, the method of Indyk and Woodruff \cite{frquency_moments1} is not directly applicable since some basic tools available for frequency vectors (such as $L_2$ norm approximation) cannot be used.
We propose a different algorithm which is still in the same spirit as \cite{frquency_moments1}; it can be found in Section \ref{sec: sdgsdgsdsdg}.
We prove Lemmas \ref{lm: estimating V} and \ref{lm: find epsilo     dfggf n} which state that an existence
of an $\alpha$-ThresholdMax algorithm for an implicitly defined vector $V$ implies an existence of
an $(\epsilon, \delta)$-approximation algorithm for $|V|$, with memory increased by an additional
factor of ${1\over\alpha}poly({1\over \epsilon}\log{nm\over\delta})$.

\bigskip
\noindent
\underline{\emph{Other Technical Issues}}.
There are several other rather technical issues that need to be resolved.
We need to prove that the methods of Indyk \cite{stable} and Indyk and McGregor \cite{mi}
are applicable for $k$-dimensional tensors that are obtained from $M_{Ind}$ by applying Prefix-Zero and Suffix-Sum operators.
The proofs can be found in Section \ref{sec:indyk genralizing}.
To prove our main theorems, certain properties of the operations on tensors
should be established. We prove these in Section \ref{sec:tensors}.

\subsection{Related Work}
Measuring pairwise independence between two or more random variables is a
fundamental problem that touches many areas of computer science.
The problems of  efficiently testing pairwise, or $k$-wise, independence were recently
considered by Alon, Andoni, Kaufman, Matulef, Rubinfeld and Xie \cite{k-wise_independenc};
Alon, Goldreich and Mansour \cite{ind3};
Batu, Fortnow, Fischer, Kumar, Rubinfeld and White \cite{batu_independence}; and
Batu, Kumar and Rubinfeld \cite{ind2}.
They addressed the problem of minimizing the number of samples needed to
obtain sufficient approximation, when the joint distribution is accessible
through a sampling procedure.
Unlike the work in \cite{k-wise_independenc, ind3, batu_independence, ind2},
in the streaming model, the joint distribution is given by a stream of tuples.

Many exciting results have been reported
in the streaming model, including, for example,
Alon, Duffield, Lund and Thorup \cite{priority1};
Alon, Matias and Szegedy \cite{ams};
Bagchi, Chaudhary, Eppstein and Goodrich \cite{epstein};
Bar-Yossef, Jayram, Kumar and Sivakumar \cite{frequency_lower_bound1};
Bar-Yossef, Kumar and Sivakumar \cite{triangles};
Beame, Jayram and Rudra \cite{beam1};
Bhuvanagiri, Ganguly, Kesh and Saha \cite{frequency_moments};
Chakrabarti, Khot and Sun \cite{set_disjointness};
Charikar, O'Callaghan and Panigrahy \cite{clustering2};
Coppersmith and Kumar \cite{frequency_impr2};
Cormode, Datar, Indyk and Muthukrishnan \cite{hamming};
Datar, Immorlica, Indyk, and Mirrokni \cite{imor};
Duffield, Lund and Thorup \cite{jacmthorup};
Feigenbaum, Kannan, McGregor, Suri and Zhang \cite{graphs};
Gal and Gopalan \cite{lis1};
Ganguly \cite{frequency_impr1};
Indyk \cite{stable};
Indyk and McGregor \cite{mi};
Indyk and Woodruff \cite{frquency_moments1};
Mitzenmacher and Vadhan \cite{mitz};
Sun and Woodruff \cite{lis};
and Szegedy \cite{dlt}.
For a detailed discussion of the streaming model, we refer readers to the excellent surveys
of Aggarwal (ed.) \cite{strbook1}; Babcock, Babu, Datar, Motwani and Widom \cite{models_issues}; and
Muthukrishnan \cite{strbook}.

In our recent work, \cite{our L2}, we also address the problem of $k$-wise independence for data stream.
In contrast to the current paper, in \cite{our L2} we study the $L_2$ norm and use entirely different techniques.


\subsection{Roadmap}
Section \ref{sec: main theorems} describes the main theorems of the paper.
In Section \ref{sec:tensors} we show some useful properties of Suffix-Sum and Prefix-Zero.
Section \ref{sec: tournaments fdffg} contains proof of the Tournament algorithm.
Section  \ref{sec: sdgsdgsdsdg} contains a generalization of the ideas of Indyk and Woodruff
\cite{frquency_moments1} to implicit vectors. Finally, Section \ref{sec:indyk genralizing} generalizes methods of Indyk \cite{stable} and Indyk and McGregor \cite{mi} to work with sampled portions of $M_{Ind}$.


\section{Main Theorems}\label{sec: main theorems}
The proof of our result is based on three main steps which are summarized by the following theorems.
The remainder of this paper is devoted to establishing these theorems.

\begin{theorem}\textbf{Dimension Reduction for Implicit Tensors }\label{tm: fff reduction}

Let $s\ge 1$ and let $M$ be a $s$-dimensional tensor with $poly(n,m)$-bounded entries that is defined by a function $\mathcal{F}$, 
i.e., $M = \mathcal{F}(D, \mathcal{H})$ where $D$ is a data stream and $\mathcal{H}$ is a fixed randomness.
Let $H:[n]\mapsto \{0,1\}$ be an arbitrary fixed hash function. 
Assume that
\begin{enumerate}
\item There exists an algorithm $\mathfrak{A}(D, \mathcal{H}, H, \delta)$ that, given $D$ and an access to $\mathcal{H}$ and $H$, in one pass obtains $(\log^k(n), \delta)$-approximation of $|W(M,H)|$;
\item There exists an algorithm $\mathfrak{B}(D, \mathcal{H}, H, \epsilon, \delta)$ that, given $D$ and an access to $\mathcal{H}$ and $H$, in one pass obtains an $(\epsilon, \delta)$-approximation of $|T_1(W(M,H))|$;
\item Both algorithm require memory $\nu(n,m,\epsilon, \delta) \le O(\left({1\over\epsilon}\log{nm\over \delta}\right)^{(30+k)^s})$,
beyond the memory required for $H$ and $\mathcal{H}$.
\end{enumerate}

Then there exists an algorithm that in one pass obtains an $(\epsilon, \delta)$-approximation of
$|M|$ using memory $\left({1\over\epsilon}\log{nm\over \delta}\right)^{(30+k)^{s+1}}$.
\end{theorem}
\begin{proof}
Follows from Theorem \ref{lm: the most importnat one tensor}, Lemma \ref{lm: estimating V}, Lemma \ref{lm: find epsilo     dfggf n} and elementary computations.

Indeed, the assumptions of Theorem \ref{sec: main theorems} imply,
by Theorem \ref{lm: the most importnat one tensor}, an existence of a ${\epsilon\over \log^{2k}(n)}$-ThresholdMax
algorithm (see Definition \ref{def: threshold algorithm dd}) for restricted function $\mathcal{F'} = AbsoluteVector(\mathcal{F}(D, \mathcal{H}))$.
The existence of a ThresholdMax algorithm implies, by Lemma \ref{lm: find epsilo     dfggf n}, 
the existence of a Cover algorithm (see Definition \ref{def: cover algorithm}) for $AbsoluteVector(\mathcal{F}(D, \mathcal{H}))$. 
The assumption that the entries of $M$ are polynomially bounded and Fact \ref{fct: bounded } imply that the entries of $AbsoluteVector(\mathcal{F}(D, \mathcal{H}))$
are polynomially bounded as well. Thus, by Lemma \ref{lm: estimating V}, there exists an $(\epsilon, \delta)$-approximation 
algorithm for $|AbsoluteVector(\mathcal{F}(D, \mathcal{H}))|$.
Finally, $|AbsoluteVector(\mathcal{F}(D, \mathcal{H}))| = \sum_{i\in [n]} |Hyperplane(\mathcal{F}(D, \mathcal{H}), l)| = |M|$.

After substituting the parameters, the memory required is less than (for sufficiently large $n$)
%
%
%
%

$$
{1\over \epsilon^{30}}\log({1\over \delta})\log^{2k+20}(nm)\nu(n,m,{\epsilon^7\over \log^4(nm)}, {\epsilon^{17}\over \log^{2k}(n)\log^8(mn)}) \le
$$
$$
\left({1\over\epsilon}\log{nm\over \delta}\right)^{(30+k)^{s+1}}.
$$
\end{proof}

\begin{theorem} \textbf{Approximation Theorem for Tensors}\label{tm: approxima theorem}

Let $M$ be a $k$-dimensional tensor with entries bounded by $poly(n,m)$ and implicitly defined by a function $\mathcal{F}(D)$. Assume that
\begin{enumerate}
\item There exists an algorithm $\mathfrak{B}_s(D, H_1,\dots, H_s)$ (for some $s< k$) that, given $D$ and an access to fixed hash functions $H_1,\dots,H_s$,  in one pass obtains an $(\epsilon, \delta)$-approximation of $|T_s(W(M,H_1,\dots, H_s))|$;
\item There exist algorithms $\mathfrak{A}_{s_1, s_2}(D, H_1,\dots, H_{s_1})$ (for any $0\le s_2 \le s_1\le s$) that, given $D$ and an access to $H_i$s, in one pass obtain a $(\log^k(n), \delta)$-approximation of $|T_{s_2}(W(M,H_1,\dots,H_{s_1}))|$;
\item All algorithms use memory bounded by $O(\left({1\over\epsilon}\log{nm\over \delta}\right)^{20})$,
beyond the memory required for $H_i$s.
\end{enumerate}

Then there exists an algorithm that in one pass obtains an $(\epsilon, \delta)$-approximation of
$|M|$ using memory $O(\left({1\over\epsilon}\log{nm\over \delta}\right)^{(30+k)^k})$.

\end{theorem}
\begin{proof}

Define $g(x) = \left({1\over\epsilon}\log{nm\over \delta}\right)^{(30+k)^{k-x}}$
First, we show that for any $s_1 \le s$ there exists an algorithm $\mathfrak{B}_{s_1}(D, H_1,\dots, H_{s_1})$ that gives an $(\epsilon, \delta)$-approximation of $|T_{s_1}(W(M,H_1,\dots,H_{s_1}))|$ and uses memory at most  $g(s_1)$.

We prove this fact by induction on $s_1$.
For $s_1 = s$, the fact follows from the first assumption of Theorem \ref{tm: approxima theorem} since $g(s) \ge \left({1\over\epsilon}\log{nm\over \delta}\right)^{20}$.
For $s_1 < s$,
denote $\mathcal{F'}(D, H_1,\dots,H_{s_1}) = T_{s_1}(W(\mathcal{F}(D), H_1,\dots,H_{s_1})$.
Denote $M' = \mathcal{F'}(D, H_1,\dots,H_{s_1})$ and let $H$ be an arbitrary hash function.
By Corollary \ref{fct:dfsvmnsvmndsvngsdbbg},
\begin{equation}\label{eq: dddddddsss}
W(M',H) = W(\mathcal{F'}(D, H_1,\dots,H_{s_1}), H) =
\end{equation}
$$
W(T_{s_1}(W(\mathcal{F}(D), H_1,\dots,H_{s_1}), H) = T_{s_1}(W(M,H_1,\dots,H_{s_1},H)).
$$
Thus, and by the second assumption of the theorem, there exists
an algorithm $\mathfrak{A}_{s_1, s_1+1}$ that in one pass obtains a $(\log^k(n), \delta)$-approximation of $|W(M',H)|$
using memory less than or equal to $g(s_1 + 1)$.

Also, by Corollary \ref{fct:dfsgfffshhhhhdbbg} and by $(\ref{eq: dddddddsss})$:
\begin{equation}\label{eq: ddddddccccdsdfsdsss}
T_1(W(M',H)) = T_1(T_{s_1}(W(M,H_1,\dots,H_{s_1},H))) = T_{s_1+1}(W(M,H_1,\dots,H_{s_1},H)).
\end{equation}
By induction, there exists an algorithm that gives an $(\epsilon, \delta)$-approximation of
$|T_{s_1+1}(W(M,H_1,\dots,H_{s'},H))| = |T_1(W(M',H))|$ using memory $g(s_1+1)$.

$M'$ is implicitly defined by a fixed function $\mathcal{F'}(D, H_1,\dots, H_s)$.
By Fact \ref{fct: bounded }, its entries are polynomially bounded.
Thus, by $(\ref{eq: dddddddsss})$ and $(\ref{eq: ddddddccccdsdfsdsss})$,
all assumptions of Theorem \ref{tm: fff reduction} are satisfied for $M'$. 
Therefore, there exists an algorithm that gives an $\epsilon$-approximation of 
$|M'| = |T_{s_1}(W(M,H_1,\dots,H_{s_1}))|$ using memory $g(s_1)$.

In particular, there exists an algorithm that for any $H$ gives an $\epsilon$-approximation of $|T_1(W(M,H))|$ using $g(1)$. Also, by the second assumption of the theorem, there exists an algorithm that
gives a $\log^k(n)$-approximation of $|T_0(W(M,H))| = |W(M,H)|$.
Thus, we can apply Theorem \ref{tm: fff reduction} for $M$ and obtain an $\epsilon$-approximation of $|M|$. 
The resulting memory usage will be $O(\left({1\over\epsilon}\log{nm\over \delta}\right)^{(30+k)^k})$.
\end{proof}

\noindent
The following lemmas are proven in Section \ref{sec:indyk genralizing}.
\begin{lemma}\label{lm: epsilon for indep tensors}
There exists an algorithm $\mathfrak{B}_{k-1}$ that, given a data stream $D$ and an
access to hash functions $H_1,\dots,H_{k-1}$, in one pass obtains an $\epsilon$-approximation
of $|T_{k-1}(W(M_{Ind},H_1,\dots, H_{k-1}))|$  using memory $O({1\over \epsilon^2}\log{1\over \delta}\log{nm\over \epsilon\delta})$.
\end{lemma}

\begin{lemma}\label{lm: polylog for indep tensors}
There exists an algorithm $\mathfrak{A}_{s_1, s_2}$ (for any $0\le s_2 \le s_1\le k$) that,
given a data stream $D$ and an access to hash functions $H_1,\dots,H_{s_1}$, in one pass
obtains a $\log^k{n}$-approximation of $|T_{s_2}(W(M_{Ind},H_1,\dots,H_{s_1}))|$ using memory $O(\log{(nm)}\log{1\over \delta})$.
\end{lemma}

\begin{theorem}\label{tm: main theorem}\textbf{Main Theorem}
Let $k\ge 2$ be a constant, and let $D$ be a stream of $k$-tuples from $[n]^k$.
For any $0<\epsilon<1$, there exists an
algorithm that makes a single pass over $D$ and returns an $(\epsilon, \delta)$-approximation of the statistical distance between product and joint distribution (see Definition \ref{def: independence problem a} ) 
using memory $O(\left({1\over \epsilon}\log({nm\over \delta})\right)^{(30+k)^k})$.
\end{theorem}
\begin{proof}

By Lemma \ref{lm: epsilon for indep tensors} and Lemma \ref{lm: polylog for indep tensors},
the algorithms required by Theorem \ref{tm: approxima theorem} exist for $M_{Ind}$.
Also, by Fact \ref{fct: bounded }, the entries of $M_{Ind}$ are polynomially bounded.
Thus all assumptions of Theorem \ref{tm: approxima theorem} are true for $M_{Ind}$.
Applying Theorem \ref{tm: approxima theorem} to $M_{Ind}$, we obtain the main result.
\end{proof}


\section{Properties of Tensors}\label{sec:tensors}

We prove the following useful facts about Suffix-Sum and Prefix-Zero operations.
\begin{fact}\label{fct:dfsgsdbbg2}
Let $M$ be a $t$-dimensional tensor and $0\le s\le t$. Then $$W(T_{s}(M),H) =  T_{s}(W(M,H_1=\one,\dots,H_s = \one,H)).$$
\end{fact}
\begin{proof}
Denote by $m_{\mathbf{w}}$ (for $\mathbf{w}\in [n]^t$) the $\mathbf{w}$-th entry of $M$.
For any $\mathbf{i}\in [n]^{t-s}$, denote by $a_{\mathbf{i}}$ the entry of $T_s(M)$. By Definition \ref{def:rollup tensor deded}:
$$
a_i = \sum_{\mathbf{j}\in [n]^s} m_{(\mathbf{j},\mathbf{i})}.
$$
Denote by $b_{\mathbf{i}}$ the entry of $W(T_s(M), H)$. By Definitions \ref{def:rollup tensor deded} and \ref{def:wider dd}:
$$
b_{\mathbf{i}} = H(\mathbf{i}_1)a_{\mathbf{i}} = \sum_{\mathbf{j}\in [n]^s} m_{(\mathbf{j},\mathbf{i})} H(\mathbf{i}_1).
$$
Denote by $c_{\mathbf{i}}$ the $\mathbf{i}$-th entry of $T_{s}(W(M,H_1=\one,\dots,H_s = \one,H))$. By Definitions \ref{def:rollup tensor deded} and \ref{def:wider dd}:
$$
c_{\mathbf{i}} = \sum_{\mathbf{j}\in [n]^s} m_{(\mathbf{j},\mathbf{i})} H(\mathbf{i}_1).
$$
Thus, for any ${\mathbf{i}}$, $b_{\mathbf{i}} = c_{\mathbf{i}}$ and the fact is correct.
\end{proof}

\begin{fact}\label{fct:dfsgvbdbsdb1bg}
Let $M$ be a $t$-dimensional tensor and let $0\le s<t$.
Then $T_1(T_{s}(M)) = T_{s+1}(M)$.
\end{fact}
\begin{proof}
Denote by $m_{\mathbf{w}}$ (for $\mathbf{w}\in [n]^t$) the $\mathbf{w}$-th entry of $M$.
For $\mathbf{j} \in [n]^{t-s}$ denote $b_{\mathbf{j}}$ to be an entry of $T_s(M)$.
By Definition \ref{def:rollup tensor deded}:
$$
b_{\mathbf{j}} = \sum_{\mathbf{u}\in [n]^{s}} m_{(\mathbf{u},\mathbf{j})}.
$$
For every $\mathbf{i}\in [n]^{t-s-1}$, denote by $c_{\mathbf{i}}$ the entry of $T_1(T_s(M))$.
By Definition \ref{def:rollup tensor deded}:
$$
c_{\mathbf{i}} = \sum_{l\in [n]} b_{(l, \mathbf{i})} = \sum_{l\in [n]} \sum_{\mathbf{u}\in [n]^{s}} m_{(\mathbf{u},(l,\mathbf{i}))} = \sum_{l\in [n]} \sum_{\mathbf{u}\in [n]^{s}} m_{((\mathbf{u},l),\mathbf{i})} = \sum_{\mathbf{v}\in [n]^{s+1}} m_{(\mathbf{v},\mathbf{i})}.
$$
For any $\mathbf{i} \in [n]^{t-s-1}$ denote by $a_{\mathbf{i}}$ the entry of $T_{s+1}(M)$. By Definition \ref{def:rollup tensor deded}:
$$
a_{\mathbf{i}} = \sum_{\mathbf{v}\in [n]^{s+1}} m_{(\mathbf{v},\mathbf{i})}.
$$
Thus, for any $\mathbf{i}$, $a_{\mathbf{i}} = c_{\mathbf{i}}$ and the fact is correct.
\end{proof}

\begin{fact}\label{fct:dfsgfffshhhhhdbbfffffff}
Let $M$ be a $t$-dimensional tensor, let $s\le t$ and let $H_1, \dots, H_s$ and $G_1,\dots, G_s$ be hush functions.
Then
$$
W(M, H_1G_1,\dots, H_sG_s) = W(W(M, H_1,\dots, H_s), G_1,\dots,G_s))
$$
\end{fact}

\begin{corollary}\label{fct:dfsvmnsvmndsvngsdbbg}
Let $M$ be a $t$-dimensional tensor and let $0\le s<t$.
Let $M' = T_{s}(W(M,H_1,\dots,H_{s}))$. Then
$$W(M',H) =  T_{s}(W(M,H_1,\dots,H_{s},H)).$$
\end{corollary}
\begin{proof}
Denote $M'' = W(M,H_1,\dots,H_{s})$. Then by Fact \ref{fct:dfsgsdbbg2}:
$$
W(M',H) = W(T_{s}(M''),H) =  T_{s}(W(M'',G_1 = \one,\dots,G_k = \one,H)).
$$
Also by Fact \ref{fct:dfsgfffshhhhhdbbfffffff}:
$$
W(M'',G_1,\dots,G_s,H) = W(W(M,H_1,\dots,H_{s}, \one),G_1,\dots,G_s,H) = W(M,H_1,\dots,H_{s},H).
$$
\end{proof}

\begin{corollary}\label{fct:dfsgfffshhhhhdbbg}
Let $M$ be a $t$-dimensional tensor and let $0\le s<t$.
Let $M' = T_{s}(W(M,H_1,\dots,H_{s}))$. Then
$$T_1(M',H)) =  T_{s+1}(W(M,H_1,\dots,H_{s},H)).$$
\end{corollary}
\begin{proof}
By Fact \ref{fct:dfsgvbdbsdb1bg} and Corollary \ref{fct:dfsvmnsvmndsvngsdbbg}:
$$T_{s+1}(W(M,H_1,\dots,H_{s},H)) = T_1(T_{s}(W(M,H_1,\dots,H_{s},H))) = T_1(W(M', H)).$$
\end{proof}

\begin{fact}\label{lm: tensor significant fff}
Let $M$ be an arbitrary $s$-dimensional tensor, let $M_l$ be $(1-\epsilon/2)$-significant hyperplane of $M$, $M_l = Hyperplane(M,l)$, and let
$M' = T_1(M)$. Then
$|M'|$ is an $\epsilon$-approximation of $|M_l|$.
\end{fact}
\begin{proof}
We have
$$
|M'| = \sum_{{\mathbf{i}} \in [n]^{s-1}}|\sum_{j\in [n]} m_{(j,{\mathbf{i}})}| \le \sum_{{\mathbf{i}} \in [n]^{s-1}}\sum_{j\in [n]} |m_{(j,{\mathbf{i}})}| = |M| \le {1\over 1-\epsilon/2}|M_l|\le (1+\epsilon)|M_l|.
$$
On the other hand,
$$
|M'| = \sum_{{\mathbf{i}} \in [n]^{s-1}}|\sum_{j\in [n]} m_{(j,{\mathbf{i}})}| \ge
\sum_{{\mathbf{i}} \in [n]^{s-1}} (|m_{(l,{\mathbf{i}})}| - \sum_{j\in [n], j\neq l}|m_{(j,{\mathbf{i}})}|) =  |M_l| - (|M| - |M_l|) =
$$
$$
\ge (2 - {1\over 1-\epsilon/2})|M_l| \ge (1-\epsilon)|M_l|.
$$
\end{proof}

\begin{fact}\label{fct: bounded } The following is correct:

\begin{enumerate}
\item Let $M$ be a $s$-dimensional tensor with polynomially bounded (in $n$ and $m$) entries for $s\le k$.
Let $M'$ be a tensor obtained from $M$ by an arbitrary composition of Prefix-Zero, AbsoluteVector, Hyperplane and Suffix-Sum operators.
Then the entries of $M'$ are polynomially bounded.
\item All entries of $M_{Ind}$ are integers with absolute values bounded by $2m^k$ and thus claim $1$ is true for $M_{Ind}$.
\end{enumerate}
\end{fact}
\begin{proof}
The first claim follows from the fact that the entries of $M'$ are sums of disjoint subsets of $M$ and that the number of entries in $M$ is bounded by $n^k$.
The second claim follows from Definition \ref{def:wider ddcccddddddffffffd}.
\end{proof}


\section{Certifying Tournaments}\label{sec: tournaments fdffg}

\bigskip
{\fbox{\begin{minipage}{6.1in}
{\begin{algorithm}\label{alg:row finder tensor} \underline{TensorTournament($D,\mathcal{H}, H,\epsilon$)}
\begin{enumerate}
\item Repeat in parallel $O({\log{1\over\delta}\over p})$ times where $p = 1-\sqrt{1-\epsilon/2}$.

\begin{enumerate}
\item Generate $2$-wise independent random hash function $Z$ from $[n]$ to $\{0,1\}$ such that $Z(i) = 0$ w.p. $0.5$. Denote
$
Z_{1} = HZ,\ \  Z_{0} = H(1-Z).
$
\item Compute in a single pass over $D$ for $i=0,1$:
$
\mathfrak{t}_{i} = \mathfrak{A}(D, \mathcal{H}, Z_{i}, \epsilon, \delta'),
$ where $\delta' = {p\epsilon\over 4\log{1\over\delta}}$.
\item Simultaneously (in the same pass), compute $
\mathfrak{l}_i = \mathfrak{B}(D, \mathcal{H}, Z_{i}, \delta').
$
\item Put
$
\mathfrak{u}_i = \max\{{\mathfrak{l}_i\over \beta}, \mathfrak{t}_{i}, 0\}, \ \ i=0,1.
$
\item Define $\lambda' = (1+\epsilon)\lambda$, where $\lambda$ is the constant from Lemma \ref{lm:x/y}, $\lambda = (1+ {2(1-\epsilon)^{1/4}\over 1 - (1-\epsilon)^{1/4}})$.
\item Compute
$$
U' =\left\{\begin{array}{ll}
        \mathfrak{u}_{1},  &\textrm{if\ } \mathfrak{u}_{1} \ge \lambda'{\beta^2} \mathfrak{u}_{0},\\
        \mathfrak{u}_{0},  &\textrm{if\ } \mathfrak{u}_{0} \ge \lambda'{\beta^2} \mathfrak{u}_{1},\\
        0, & \textrm{\ otherwise}.\\
    \end{array}\right.
$$
\end{enumerate}

\item Output $U$ to be the minimum of all $U'$s.

\end{enumerate}
\end{algorithm}
}
\end{minipage}}}

\bigskip
\noindent

\begin{definition}\label{def: threshold algorithm dd}
Let $\mathcal{F}$ be a fixed function that defines implicit vectors, given a data stream and a fixed randomness 
and denote $V = \mathcal{F}(D, \mathcal{H})$ as a vector with entries $v_i$.
For $\alpha > 0.5$, an \emph{$\alpha$-ThresholdMax algorithm for restricted $\mathcal{F}$} is
an algorithm that receives as an input a data stream $D$ and an access to a randomness $\mathcal{H}$ and a random function $H: [n]\mapsto \{0,1\}$, and in one pass over $D$ returns $U\ge 0$ such that w.p. at least $1-\delta$:
\begin{enumerate}
\item If $U > 0$ then $U$ is an $\epsilon$-approximation of $|v_i|$ for some $i$ with $H(i) = 1$.
\item If\footnote{Here and thenceforth we denote by $VH$ a vector with entries $v_iH(i)$, $i\in [n]$ } $|VH| > 0$ and there exists $i$ such that $H(i) = 1$ and $|v_i| \ge (1-\alpha)|VH|$ then $U$ is an $\epsilon$-approximation of $|v_i|$.
\end{enumerate}
\end{definition}

\begin{theorem}\label{lm: the most importnat one tensor}
Let $H$ be a fixed hash function defined as above and let $\epsilon \le 0.1$.
Let $M$ be a $s$-dimensional tensor implicitly defined by a fixed function $\mathcal{F}$, stream $D$ and randomness $\mathcal{H}$, $M = \mathcal{F}(D, \mathcal{H})$.
If there exist:
\begin{itemize}
\item An algorithm $\mathfrak{A}(D, \mathcal{H}, H, \delta)$ that in one pass obtains $(\beta, \delta)$-approximation of $|W(M,H)|$ using memory $\mu_1(n,m,\epsilon, \delta)$;
\item  An algorithm $\mathfrak{B}(D, \mathcal{H}, H, \epsilon, \delta)$ that in one pass over $D$ obtains an $(\epsilon, \delta)$-approximation of $|T_1(W(M,H))|$ using memory $\mu_2(n,m,\epsilon, \delta)$;
\end{itemize}
Let $\alpha = {\epsilon\over 64\beta^2}$.
Then the Algorithm $TensorTournament(D,\mathcal{H}, H, \epsilon)$ is an $\alpha$-ThresholdMax algorithm for restricted $\mathcal{F'}$ (see Definition \ref{def: threshold algorithm dd}), where $\mathcal{F'}(D, \mathcal{H}) = AbsoluteVector(\mathcal{F}(D, \mathcal{H}))$.
The algorithm makes a single pass over $D$ and uses memory
$$
O({1\over \epsilon}\log{1\over \delta}(\mu_1(n,m,\epsilon/3, {\delta\epsilon/ \log{(1/\delta)}}) + \mu_2(n,m,\epsilon/3, {\delta\epsilon/ \log{(1/\delta)}}) + \log{nm}).
$$
\end{theorem}
\begin{proof} \

\noindent
Denote $M^{t} = W(M, Z_{t})$ for $t=0,1$.
Let $M_i = Hyperplane(M, i)$ for $i\in [n]$ and let $V'$ to be a vector with elements
$
|M_{i}|.
$
By Definition \ref{def:rollup tensfsssssor deded}, $V' = \mathcal{F'}(D, \mathcal{H})$.
Further, let $V$ be a vector with entries $v_i = |M_i|H(i)$.
We prove that the algorithm satisfies two conditions of Definition \ref{def: threshold algorithm dd} for the ThresholdMax algorithm for $V$ and $H$.
\

\

\noindent
\emph{\textbf{Proof of the first condition of Definition \ref{def: threshold algorithm dd}}}
\

\noindent
We prove the following stronger statements which imply the first condition of Definition \ref{def: threshold algorithm dd}:
\begin{enumerate}[I.]
\item If there is no $(1-\epsilon)$-significant entry $v_l$ then, w.p. at least $1-{\delta\over 3}$, $U=0$.
\item If $|V|>0$ and there is a $(1-\epsilon)$-significant entry $v_l$ then, w.p. at least $1-{\delta\over 3}$,
either $U=0$ or $U$ is a $3\epsilon$-approximation of $|v_l|$.
\end{enumerate}

\

\noindent
\textbf{Proof of statement $\bf{I}$}

\noindent
By definitions of $\mathfrak{B},\mathfrak{A}$, we have w.p. at least $1-{8\delta'}$ for $t=0,1$:
$
\mathfrak{u}_t \ge {\mathfrak{l}_t\over \beta} \ge {|M^t|\over \beta^2};
$
and
$\mathfrak{t}_{t} \le (1+\epsilon)|T_1(M^{t}, H)| \le (1+\epsilon)|M^{t}|$;
and
${\mathfrak{l}_t\over \beta} \le {|M^t|}$.
Thus,
\begin{equation}\label{eq:gfrgd}
{|M^t|\over \beta^2} \le \mathfrak{u}_t \le (1+\epsilon)|M^t|.
\end{equation}

Following the terminology of Lemma \ref{lm:x/y}, we define $X = |M^{1}|$ and $Y = |M^{0}|$.
We have the following relations:
$$
|V| = \sum_{i}v_i = \sum_{i}H(i)|M_i| = \sum_{i\in[n]}H(i)\sum_{\mathbf{j}'\in [n]^{s-1}}|m_{(i,\mathbf{j}')}| = |W(M,H)|,
$$
$$
X = |M^1| = \sum_{\mathbf{j}\in [n]^s}Z(\mathbf{j}_1)H(\mathbf{j}_1)|m_\mathbf{j}| = \sum_{i\in[n]}Z(i)H(i)\sum_{\mathbf{j}'\in [n]^{s-1}}|m_{(i,\mathbf{j}')}| = \sum_{i}Z(i)H(i)|M_i| = \sum_{i}Z(i)v_i,
$$
and similarly
\begin{equation}\label{eq:yff}
Y = |M^{0}| = \sum_{i}(1-Z(i))H(i)|M_i| = |V| - X = |V| - |M^{1}|.
\end{equation}
By statement $\bf{I}$, for all $i$,
$
v_i < (1-\epsilon)|V|.
$
Thus we can apply Lemma \ref{lm:x/y}.  We have:
$$
P((|M^{0}| \ge \lambda|M^{1}|) \cup (|M^{1}| \ge \lambda|M^{0}|) ) = P((X \ge \lambda Y) \cup (Y \ge \lambda X) ) \le \sqrt{1-\epsilon}.
$$
Let $\Upsilon$ be the event $(\mathfrak{u}_{0} \ge \lambda'\beta^2 \mathfrak{u}_{1}) \cup (\mathfrak{u}_{1} \ge \lambda'\beta^2 \mathfrak{u}_{0})$. Let $\Phi$ be the event that ${|M^t|\over \beta^2} \le \mathfrak{u}_t \le (1+\epsilon)|M^t|
$ for both values of $t$.
We have
$
P( \Upsilon) \le P( \Upsilon, \Phi) + P(\bar{\Phi}).
$
By $(\ref{eq:gfrgd})$, we have
$
P(\bar{\Phi}) \le {8\delta'}.
$
Also, events $\mathfrak{u}_{0} \ge \lambda'\beta^2 \mathfrak{u}_{1}$ and $\Phi$
imply that $|M^{0}| \ge \lambda|M^{1}|$; indeed:
$$
|M^{0}| \ge {\mathfrak{u}_{0} \over (1+\epsilon)} \ge {\lambda'\over 1+\epsilon}\beta^2\mathfrak{u}_{1} \ge \lambda|M^{1}|.
$$
Thus we have
$$
P( \Upsilon, \Phi) \le
P((|M^{0}| \ge \lambda|M^{1}|)\cup (|M^{1}| \ge \lambda|M^{0}|)) \le \sqrt{1-\epsilon}.
$$
We summarize that if no $(1-\epsilon)$-significant $v_i$ exists, then
$$
P(U' \neq 0) \le P(\Upsilon) \le \sqrt{1-\epsilon} + O(\delta')\le \sqrt{1-\epsilon/2}.
$$
Recall that the number of repetitions is $O({1\over p}\log{ 1/\delta})$, where $p = 1-\sqrt{1-\epsilon/2}$.
Thus
$
P(U \neq 0) \le (1-p)^{{1\over p}\log{3\over\delta}} \le {\delta\over 3}.
$

\
\

\noindent
\textbf{Proof of statement $\bf{II}$}

Let $v_l$ be a $(1-\epsilon)$-significant entry of $V$.
Assume, w.l.o.g., that for one execution of the main cycle of the $Tournament$ algorithm, $Z(l) = 0$.
Statement $\bf{II}$ implies $|V| > 0$ which implies $v_l = |M_l|H(l) > 0$ which implies $(1-Z(l))H(l) = 1$.
Thus, $|Hyperplane(M^0, l)| = |Hyperplane(W(M, (\one-Z)H), l)| = |M_l| = v_l$.
Therefore by $(\ref{eq:yff})$, $|Hyperplane(M^0, l)| = v_l \ge (1-\epsilon)|V| \ge (1-\epsilon)|M^{0}|$, i.e.,
the $l$-th hyperplane of $M^0$ is $(1-\epsilon)$-significant.
By Fact \ref{lm: tensor significant fff}, $|T(M^{0})|$ is an
$2\epsilon$-approximation of $|M_l|$. By the assumptions of the theorem, $\mathfrak{B}$ returns an $\epsilon$-approximation of $|T(M^{0})|$.
Thus, $\mathfrak{t}_0$ is a $3\epsilon$-approximation of $|M_l|$, w.p. at least $1-\delta'$, in which case
$$
\mathfrak{u}_0 \ge \mathfrak{t}_0\ge (1-3\epsilon)|M_l|.
$$
Also, by the assumption of Theorem \ref{lm: the most importnat one tensor}, w.p. at least $1-\delta'$, we have ${\mathfrak{l}_0 \over \beta} \le |M^{0}|$. Thus
$$
\mathfrak{u}_0 = \max\{{\mathfrak{l}_0\over \beta}, \mathfrak{t}_0, 0\}
        \le \max\{|M^{0}|, (1+3\epsilon)|M_l|\} \le (1+3\epsilon)|M_l|.
$$
On the other hand, w.p. at least $1-{2\delta'}$
$$
\mathfrak{u}_1 = \max\{{\mathfrak{l}_1 \over \beta}, \mathfrak{t}_1, 0\}
        \le \max\{|M^{1}|, (1+\epsilon)|M^{1}|\} = (1+\epsilon)|M^{1}|.
$$
But since $Z_s(l) = 0$ we have by $(\ref{eq:yff})$:
$$
|M^1| = |V| - |M^0| \le |V| - |Hyperplane(M^0, l)| = |V| - v_l \le {\epsilon \over 1-\epsilon}|M_l|.
$$
Combining all of the above computations, we conclude that w.p. at least $1-4\delta'$ (for sufficiently small $\epsilon$, e.g., $\epsilon \le 0.1$):
$$
\mathfrak{u}_1 \le (1+\epsilon)|M_{1}| \le {\epsilon (1+\epsilon)\over 1-\epsilon}|M_l| \le {\epsilon (1+\epsilon)\over (1-\epsilon)(1-3\epsilon)} \mathfrak{u}_0 < \lambda'\mathfrak{u}_0.
$$
Thus, $U'$ is equal to either $0$ or $\mathfrak{u}_0$ w.p. at least $1-4\delta'$. Recall simultaneously $\mathfrak{u}_0$ is a $3\epsilon$-approximation of $|M_l| = v_l$.
The same inequality is true if $Z(l) = 1$.
By union bound, w.p. at least $1- \Omega({\log{1\over \delta}\over p}\delta') = 1- \Omega({\delta})$, $U$ is either $0$ or a $3\epsilon$-approximation of $v_l$.

\
\

\noindent
\emph{\textbf{Proof of the second condition of Definition \ref{def: threshold algorithm dd}}}

Finally, consider the case when $v_l$ is a $(1-\alpha)$-significant entry of $V$.
Consider the case when $Z(l) = 0$. Repeating the arguments from the proof of statement $\bf{II}$, we have,
w.p. at least $1-4\delta'$,
$
\mathfrak{u}_0
$
is a $3\epsilon$-approximation of $v_l$ and
$$
\mathfrak{u}_1 \le (1+\epsilon)|M^{1}| \le (1+\epsilon){\alpha \over (1- {\alpha})}v_l \le 4\alpha v_l.
$$
Therefore,
$$
\mathfrak{u}_0 \ge (1-3\epsilon)v_l \ge {(1-3\epsilon)\over 4\alpha}\mathfrak{u}_1 \ge \lambda'\beta^2 \mathfrak{u}_1.
$$
Thus, w.p. $1-4\delta'$, $U' = \mathfrak{u}_0 = (1\pm 3\epsilon) v_l$.
The same is true when $Z(l) = 1$.
Thus,
$U$ is a $3\epsilon$-approximation of $v_l$ w.p. at least $ 1- \Omega({\delta})$.

\
\

\noindent
\emph{\textbf{Conclusion and memory analysis}}

Since both conditions of Definition \ref{def: threshold algorithm dd} are met (substituting $\epsilon$ with $\epsilon/3$), we conclude that $TensorTournament$ is an $\alpha$-ThresholdMax algorithm for restricted $\mathcal{F'}$.
Let us count the memory needed for a single iteration of the main cycle of the algorithm.
To generate pairwise independent $Z$, we need $O(\log{n})$ bits. In addition, we need
$\mu_1 + \mu_2$ for the algorithms $\mathfrak{B}$ and $\mathfrak{A}$ and $O(\log{nm})$ bits to keep the auxiliary variables.
Thus, in total we need memory
$$
O({1\over \epsilon}\log{1\over \delta}(\mu_1(n,m,\epsilon/3, {\delta\epsilon/ \log{(1/\delta)}}) + \mu_2(n,m,\epsilon/3, {\delta\epsilon/ \log{(1/\delta)}}) + \log{nm}).
$$
Recall that we do not count memory required to store $\mathcal{H}$ and $H$.
\end{proof}

\begin{lemma}\label{lm:x/y}
Let $V$ be a $n$-dimensional vector with non-negative entries $v_i \ge 0, i\in [n]$.
Let $Z$ be $2$-wise independent random hash functions from  $[n]$ to $\{0,1\}$, such that $P(Z(i) = 1) = 0.5$.
Let $X = \sum_i v_iZ(i)$, and  $Y = L_1(V) - X$.
If there exists $\epsilon > 0$ such that for all $i$ $v_i< (1-\epsilon) L_1(V)$, then for $\lambda = \lambda(\epsilon)\ge 1+ {2(1-\epsilon)^{1/4}\over 1 - (1-\epsilon)^{1/4}}$ we have
$$
P((X \ge \lambda Y) \cup (Y \ge \lambda X) ) \le \sqrt{1-\epsilon}.
$$
\end{lemma}
\begin{proof}
Clearly, $E(X) = L_1(V)/2$.
Further, by $2$-wise independency of $Z$, we have
$$
E(X^2) = E((\sum_i v_iZ(i))^2) = {1\over 2}\sum_i v^2_i + {1\over 4}\sum_{i\neq j} v_iv_j = {1\over 4}\sum_i v^2_i + E(X)^2.
$$
Thus, by the assumption that $v_i< (1-\epsilon) L_1(V)$, we have:
$$
Var(X) = E(X^2) - E(X)^2 = {1\over 4}\sum_i v_i^2 \le {1-\epsilon\over 4}L_1(V)^2.
$$
Thus, $\sigma(X) \le {\sqrt{1-\epsilon}\over 2}L_1(V)$.
Note that event $X\ge \lambda Y$ is equivalent to the event $X - E(X) \ge {\lambda-1\over 2(\lambda+1)}L_1(V)$ and
event $Y \ge \lambda X$ is equivalent to the event  $E(X) - X \ge {\lambda-1\over 2(\lambda+1)}L_1(V)$. Thus

$$
P((X \ge \lambda Y) \cup (Y \ge \lambda X) ) = P(|E(X) - X |\ge {\lambda-1\over 2(\lambda+1)}L_1(V)) \le
$$
$$
P(|E(X) - X |\ge {\lambda-1\over 2(\lambda+1)}{2\over  {\sqrt{1-\epsilon}}}\sigma(Y)) \le
$$
$$
(1-\epsilon)\left({\lambda+1\over \lambda-1}\right)^2 \le \sqrt{1-\epsilon}.
$$
for $\lambda \ge 1+ {2(1-\epsilon)^{1/4}\over 1 - (1-\epsilon)^{1/4}}$.

Note that if there is at most one strictly positive $v_i$, then $P((X \ge \lambda Y) \cup (Y \ge \lambda X) ) = 1$ seems to contradict our
lemma. However, in this case, there exists $v_i$, such that $v_i = L_1(V)$, and thus the assumption of the lemma does not hold.
Generally, the assumptions imply that there exists at least $1\over 1-\epsilon$ strictly positive entries $v_i$.
\end{proof}


\section{Approximating $L_1$ Norms of Implicit Vectors}\label{sec: sdgsdgsdsdg}


\begin{definition}\label{def: cover ddddff}
Let $V$ with $v_i\ge 0$ be a vector from $R^n$. A set $\mathcal{U}$ of positive numbers is an \emph{$\epsilon$-cover} of $V$ if:
\begin{enumerate}
\item All elements of $\mathcal{U}$ are $\epsilon$-approximations of distinct and positive coordinates from $V$. I.e., there is a one-to-one mapping $\rho$ from the set $\mathcal{U}$ to a subset $S' \subseteq [n]$ such that for all $U \in \mathcal{U}$, $U$ is an $\epsilon$-approximation of $v_{\rho(U)}$.
\item $\mathcal{U}$ contains $\epsilon$-approximations of all $\epsilon$-significant elements of $V$. I.e., for all $v_i$ such that $v_i \ge \epsilon |V|$, it is true that $i\in S'$.
\end{enumerate}
The \emph{size} of the cover is $|\mathcal{U}|$.
\end{definition}

\begin{definition}\label{def: cover algorithm}
Let $\mathcal{F}$ be a fixed function that implicitly defines vectors, given a data stream $D$ and a fixed randomness $\mathcal{H}$. Denote $V = \mathcal{F}(D, \mathcal{H})$.
A \emph{Cover algorithm for restricted $\mathcal{F}$} is
an algorithm that receives as an input a data stream $D$, an access to a randomness $\mathcal{H}$ and a random function $H: [n]\mapsto \{0,1\}$ and an $\epsilon$ and $\delta$.
The algorithm makes a single pass over $D$ and w.p. at least $1-\delta$, returns an $\epsilon$-cover of vector with entries $v_iH(i)$.
\end{definition}

\subsection{Witnessing $\bf{\epsilon}$-Significant Hyperplanes}\label{sec:epsilon sign}

\begin{lemma}\label{lm: find epsilo     dfggf n}
Let $\mathcal{F}$ be a fixed function that implicitly defines vectors, given a data stream $D$ and a fixed randomness $\mathcal{H}$.
An existence of $\alpha$-ThresholdMax algorithm for restricted $\mathcal{F}$ that uses memory $\mu(n,m, \epsilon, \delta)$ implies an existence of
a Cover algorithm for restricted $\mathcal{F}$ for any $\epsilon$. The Cover algorithm uses memory $O({1\over \epsilon^2\delta\alpha}(\mu(n,m, \epsilon,\delta^2\epsilon^2\alpha) + \log{nm}))$.
\end{lemma}

\begin{proof}

Denote by $\mathfrak{L}_{\alpha}(D, \mathcal{H}, H, \epsilon, \delta)$ the existing $\alpha$-ThresholdMax algorithm for restricted $\mathcal{F}$.

Using $\mathfrak{L}_{\alpha}$ we construct the following algorithm.
Let $\epsilon' = \epsilon^2\delta/3$ and $\varrho = \lceil{1\over \epsilon'\alpha}\rceil$.
Let $G$ be a pairwise independent random hash function from $[n]$ to $[\varrho]$ that is independent of $\mathcal{H}$ and $H$.
For $s\in [\varrho]$, define function $F_s$ as $F_s(i) = \one_{G(i) = s}$ and execute, in parallel for all $s$, $\mathfrak{L}_{\alpha}(D,\mathcal{H}, HF_{s},\epsilon, \delta/\varrho$).
Let $U_s$ be the output of $s$-th ran of $\mathfrak{L}_{\alpha}$. The output of our new algorithm is a set of all strictly positive $U_s$. We show below that the output is indeed $\epsilon$-cover of $V$ with probability at least $1-\delta$.

Let $V = \mathcal{F}(D, \mathcal{H})$ be a vector with entries $v_i$ and let $V_s$ be a vector with entries $v_{s,i} = v(i)F_s(i)$.
By the union bounds and by the definition of $\alpha$-ThresholdMax algorithm, w.p. at least
$1-\delta$, every positive $U_s$ is an $\epsilon$ approximation of $|v_{i_s}|$ for some $i_s$ with $H({i_s})F_s(i_s) = 1$. But this implies that $U_s$ is an approximation of $|v_i|$ with $H(v_i) = 1$.
Since $G$ splits $[n]$ into disjoint subsets, the output of our algorithm corresponds to $\epsilon$-approximations of absolute values of a
set of distinct entries of $V$. I.e., the first condition of $\epsilon$-cover is correct.

To show that the second condition is true as well, let $S_{\epsilon}$ be set of all
$i$s such that $|v_iH(i)| \ge \epsilon |VH| > 0$.
Consider a fixed $i\in S_{\epsilon}$. Let
$$
X_i = |VHF_{G(i)}| - |v_i| = \sum_{j\neq i} |v_j| H(j)F_{{G(i)}}(j) \ge 0.
$$
By pairwise independency of $G$:
$$
E(X_i) = \sum_{j\neq i} |v_j| H(j)P(G(j) = G(i)) \le {|VH| \over \varrho}.
$$
Let $\Psi_i$ be the event that
$
X_i > {\epsilon\over \varrho\epsilon'}|VH|;
$
by Markov inequality
$
P(\Psi_l) \le {\epsilon'\over\epsilon}.
$
Note that if $\Psi_i$ does not happen, then
$$
|VHF_{G(i)}| - |v_i| \le {\epsilon\over \varrho\epsilon'}|VH| \le {1\over \varrho\epsilon'}|v_i|\le \alpha  |v_i|,
$$
in which case $|v_i| \ge (1-\alpha)|VHF_{G(i)}|$.
Let $\Gamma_l$ be the event that $U_{G(i)}$ is not an $\epsilon$-approximation of $|v_i|$.
By the properties of algorithm $\mathfrak{L}_{\alpha}$,
$
P(\Gamma_i | \bar{\Psi}_i) \le {\delta\over \varrho}.
$
Thus
$$
P(\Gamma_i) \le P(\Gamma_i | \bar{\Psi}_i) + P({\Psi}_i) \le {\delta\over \varrho} + {\epsilon'\over \epsilon}.
$$
Finally, let $\Phi_{i,j}$ be the event where there is a collision between $i$ and $j$.
By pairwise independence of $G$, $P(\Phi_{i,j}) = {1\over \varrho}$, and thus the probability of collisions for $\epsilon$-significant entries is bounded by ${1\over \epsilon^2\varrho}$.
Thus, the probability that the output of the algorithm does not meet the second condition of $\epsilon$-cover is bounded by
$$
P((\cup_{i\in S_{\epsilon}}\Gamma_i)\cup (\cup_{i,j\in S_{\epsilon}}\Phi_{i,j})) \le {\delta\over \varrho\epsilon} + {\epsilon'\over \epsilon^2} + {1\over \varrho\epsilon^2} \le \delta.
$$
\end{proof}

\subsection{The $\epsilon$-Approximation}\label{sec: genralizing IW}
\begin{definition}\label{def: fsssssssfggv}
Let $\mathcal{F}$ be a fixed function that defines an implicit vector $V = \mathcal{F}(D, \mathcal{H})$, given $D$ and a randomness $\mathcal{H}$,
as in Definition \ref{def:wider ddcccddddddd}.
An algorithm that receives as an input a data stream $D$ and an access to a randomness $\mathcal{H}$ and in one pass over $D$ returns an $(\epsilon, \delta)$-approximation of $|\mathcal{F}(D, \mathcal{H})|$ is called an \emph{$(\epsilon, \delta)$-approximation algorithm for $L_1(\mathcal{F})$}.
\end{definition}

The main goal of this section is to prove
\begin{lemma}\label{lm: estimating V}
Let $\mathcal{F}$ be a fixed function that defines an implicit vector $V = \mathcal{F}(D, \mathcal{H})$, given $D$ and a randomness $\mathcal{H}$.
Assume that $V$ has non-negative entries bounded by $poly(n,m)$.
Then the existence of Cover algorithm $\mathfrak{Q}(D, \mathcal{H}, H, \epsilon, \delta)$ for restricted $\mathcal{F}$ (see Definition \ref{def: cover algorithm}) that uses memory $\mu(n,m,\epsilon,\delta)$ implies an existence of
an $(\epsilon, 2/3)$-approximation algorithm for $L_1(\mathcal{F})$ (Definition \ref{def: fsssssssfggv}) that uses memory
$$
O\left({1\over \epsilon}\log(n)\mu(n,m,{\epsilon^7\over \log^3(nm)},{\epsilon \over \log(nm)}) + {1\over \epsilon^2}\log^2(nm)\right).
$$
\end{lemma}

\subsubsection{Notations}
In this section, let $0< \epsilon < 1$ be a constant,
Define
$$
\mathfrak{a} = O(\log_{(1+\epsilon)}{n}),\ \ \ \mathfrak{b} = O(\log_{(1+\epsilon)}{nm}), \ \ \ {\chi}' = 10(\mathfrak{a} + \mathfrak{b}),\ \ \ {\chi} = \lceil{16\over \epsilon^3}{\chi}'\rceil,
$$
$$
Q = \lceil{20{\chi}'\over \epsilon^2}\rceil, \ \ \ \zeta = (1+\epsilon)^{1/Q} - 1,\ \ \ \mathfrak{c} = \min\{{\zeta\over 2(1+\zeta)}, {\epsilon\over 4\chi\chi'^2}\}.
$$
For $x>{\chi}$, let ${\mathfrak{f}}_{\chi}(x)$ be an integer such that ${\chi} (1+\epsilon)^{{\mathfrak{f}}_{\chi}(x)-1} < x \le {\chi} (1+\epsilon)^{{\mathfrak{f}}_{\chi}(x)}$ i.e., ${\mathfrak{f}}_{\chi}(x) = \lceil\log_{(1+\epsilon)}{x\over {\chi}}\rceil$.
It is easy to see that for $x>{\chi}$ we have  ${\mathfrak{f}}_{\chi}(x)>0$.

\subsubsection{Technical Lemmas}

Let $n_{rest}$ be such that $n \ge n_{rest} > {\chi}$.
For any $j = [\mathfrak{a}]$ and for every $i\in [n_{rest}]$, let $X_{i,j}$ be pairwise independent zero-one random variables with $P(X_{i,j} = 1) = {1\over(1+\epsilon)^j}$.
Let $Y_j = \sum_{i\in [n_{rest}]} X_{i,j}$.

\begin{fact}\label{lm: ddd upper bounds layers}
$$
P({{\chi} \over (1+\epsilon)^2} < Y_{{\mathfrak{f}}_{\chi}(n_{rest})} \le (1+3\epsilon){\chi}) \ge 1 - {1\over {\chi}'}.
$$
\end{fact}
\begin{proof}
Let $j_0 = {\mathfrak{f}}_{\chi}(n_{rest})$; note that $j_0 > 0$ since $n_{rest}>{\chi}$.
We have
$
E(Y_{j_0}) = {n_{rest}\over (1+\epsilon)^{j_0}}.
$
and, by pairwise independency of $X_{j_0, i}$:
$$
Var(Y_{j_0}) = n_{rest}Var(X_{j_0,1}) = {n_{rest}\over (1+\epsilon)^{j_0}}(1-{1\over (1+\epsilon)^{j_0}}) \le {n_{rest}\over (1+\epsilon)^{j_0}} \le {\chi}.
$$
Let $\epsilon' = {\epsilon \over 2}$; we have, by Chebyshev's inequality:
$$
P\left(|Y_{j_0} - {n_{rest}\over (1+\epsilon)^{j_0}}| \ge \epsilon'{n_{rest}\over (1+\epsilon)^{j_0}}\right) \le
Var(Y_{j_0})\left({(1+\epsilon)^{j_0}\over \epsilon' n_{rest} }\right)^2 \le
$$
$$
{1\over \epsilon'^2}{\chi}({(1+\epsilon)^{j_0}\over n_{rest}})^2\le
{1\over \epsilon'^2}{(1+\epsilon)^2\over {\chi}} \le {1\over {\chi}'}.
$$
Also,
$
|Y_{j_0} - {n_{rest}\over (1+\epsilon)^{j_0}}| < \epsilon'{n_{rest}\over (1+\epsilon)^{j_0}}
$
implies
$$
Y_{j_0} < (1+\epsilon'){n_{rest}\over (1+\epsilon)^{j_0}} \le (1+3\epsilon){\chi},
$$
and
$$
Y_{j_0} > (1-\epsilon'){n_{rest}\over (1+\epsilon)^{j_0}} \ge (1-\epsilon'){{\chi}\over (1+\epsilon)} \ge  {{\chi}\over (1+\epsilon)^2}.
$$
\end{proof}

\begin{fact}\label{lm: chernoff upper bounds layers}
Let $Z = \max_{j\in [\mathfrak{a}]} \{j: {{\chi} \over (1+\epsilon)^2} < Y_j \le (1+3\epsilon){\chi}\}$ if at least one such $j$ exists and $0$ otherwise. Then
$$
P(Z > {\mathfrak{f}}_{\chi}(n_{rest}) + 2) \le 1 - {1\over {\chi}'}.
$$
\end{fact}
\begin{proof}
Let $j_0 = {\mathfrak{f}}_{\chi}(n_{rest})$ and consider fixed $j' > j_0 + 2$. We have,
$$
E(Y_{j'}) = {n_{rest}\over (1+\epsilon)^{j'}},
$$
and by pairwise independency of $X$s
$$
Var(Y_{j'}) = n_{rest}Var(X_{j',1}) = {n_{rest}\over (1+\epsilon)^{j'}}(1-{1\over (1+\epsilon)^{j'}}) \le {n_{rest}\over (1+\epsilon)^{j'}} \le {{\chi}\over (1+\epsilon)^{j'-j_0}}.
$$
Thus,
$$
P(Y_{j'} > {{\chi}\over (1+\epsilon)^2}) = P(Y_{j'} -E(Y_{j'})> {{\chi}\over (1+\epsilon)^2}- E(Y_{j'})) \le {Var(Y_{j'}) \over ({{\chi}\over (1+\epsilon)^2}- E(Y_{j'}))^2} \le
$$
$$
{{\chi}\over (1+\epsilon)^{j'-j_0} ({{\chi}\over (1+\epsilon)^2}- {n_{rest}\over (1+\epsilon)^{j'}})^2} \le
{{\chi}\over (1+\epsilon)^{j'-j_0} ({{\chi}\over (1+\epsilon)^2}-  {{\chi}\over (1+\epsilon)^{j'-j_0}})^2} =
$$
$$
{1\over {\chi}}{1\over (1+\epsilon)^{j'-j_0} ({1\over (1+\epsilon)^2}-  {1\over (1+\epsilon)^{j'-j_0}})^2} \le {1\over {\chi}}{1\over (1+\epsilon)^{j'-j_0} ({1\over (1+\epsilon)^2}-  {1\over (1+\epsilon)^3})^2} \le
$$
$$
{1\over {\chi}}{1\over (1+\epsilon)^{j'-j_0} {\epsilon^2\over (1+\epsilon)^6}} \le {(1+\epsilon)^3\over \epsilon^2{\chi}}{1\over (1+\epsilon)^{j'-j_0-3}}.
$$
Clearly $Z=j'$ implies $Y_{j'} > {{\chi}\over (1+\epsilon)^2}$.
Thus, and by union bound over all $j' \ge j_0 + 3$, we have that
$$
P(Z > j_0 + 2) \le {(1+\epsilon)^3\over \epsilon^2{\chi}}\sum_{j'= j_0+3}^{\mathfrak{b}}{1\over (1+\epsilon)^{j'-j_0-3}} \le {(1+\epsilon)^3\over \epsilon^2{\chi}}{(1+\epsilon)\over \epsilon} \le {1\over {\chi}'}.
$$
\end{proof}

\begin{corollary}\label{cor:Z'}
Let $Y'_i = \sum_{j\in [n_{rest}]} \alpha_{i,j}X_{i,j}$, where $\alpha_{i,j}$ are  arbitrary random zero-one variables. 
For $Z' = \max_{j\in [\mathfrak{a}]} \{j: {{\chi} \over (1+\epsilon)^2} < Y'_j \le (1+3\epsilon){\chi}\}$, it is true that
$P(Z' > {\mathfrak{f}}_{\chi}(n_{rest}) + 2) \le {1\over {\chi}'}$.
\end{corollary}
\begin{proof}
We have for any $j$
$$
P(Z' = j) \le P(Y'_{j} > {{\chi}\over (1+\epsilon)^2}) \le P(Y_{j} > {{\chi}\over (1+\epsilon)^2}).
$$
Thus, we can repeat the arguments from Fact \ref{lm: chernoff upper bounds layers}.
\end{proof}

\begin{fact}\label{fact:zeta}
Let $\zeta = (1+\epsilon)^{1/Q} - 1$, then $
\zeta \ge {\epsilon\over 2Q}.
$
\end{fact}
\begin{proof}
If $\zeta < {\epsilon\over 2Q}$, then we have
$$
(1+\zeta)^{Q} - 1 = \sum_{i = 1}^Q \zeta^i{Q\choose i} \le \sum_{i = 1}^Q \zeta^iQ^i < \sum_{i = 1}^Q {\epsilon^i\over 2^iQ^i}Q^i = \sum_{i = 1}^Q {\epsilon^i\over 2^i} \le \epsilon.
$$
Thus, it must be the case that $\zeta \ge {\epsilon\over 2Q}$.
\end{proof}

\subsubsection{The Algorithm and Proof of Lemma \ref{lm: estimating V}}

\bigskip
{\fbox{\begin{minipage}{6.1in}
{\begin{algorithm}\label{alg:approx zuzik finder} \underline{$\mathfrak{G}(D,\mathcal{H},\epsilon, \delta$)}
\begin{enumerate}

\item Pick random integer $q$ from $0,\dots, Q-1$.

\item For any $j  \in [\mathfrak{a}]$ generate pairwise-independent random hash functions $G_j:[n] \rightarrow \{0,1\}$ such that for any $i\in [n]$ $P(G_j(i) = 1) = {1\over (1+\epsilon)^j}$.

\item In parallel, apply $\mathfrak{Q}_j = \mathfrak{Q}(D, \mathcal{H}, G_j, \mathfrak{c}, {1\over {\chi}'})$ for all $j = 0,\dots, \mathfrak{a}$.


\item For all $0\le j\le \mathfrak{a}$ and all $l =-1,\dots,\mathfrak{b}$ compute $\mathcal{Y}_{l,j}$ that is a number of elements returned by $\mathfrak{Q}_j$ in the range $[(1+\zeta)^q(1+\epsilon)^l, (1+\zeta)^q(1+\epsilon)^{l+1})$.

\item For every $l \in [\mathfrak{b}]$ compute $\mathcal{Z}_l = \max_{j>0} \{j: {{\chi} \over (1+\epsilon)^2} < \mathcal{Y}_{l,j} \le  (1+3\epsilon){\chi}\}$; 
define $\mathcal{Z}_l = 0$ if no such $j$ exists.



\item Return $(1+\zeta)^q\sum_{l\in [\mathfrak{b}]} (1+\epsilon)^{\mathcal{Z}_l+l}\mathcal{Y}_{l,\mathcal{Z}_l}$.

\end{enumerate}
\end{algorithm}
}
\end{minipage}}}

\bigskip
\noindent

Let $\mathcal{F}$ be a fixed function that defines vector $V = \mathcal{F}(D, \mathcal{H})$ with non-negative entries $v_i$ such that $L_{\infty}(V) = poly(n,m)$.
Define $q$ to be a uniform random integer from $0,\dots,Q-1$. For $l = -1, \dots, \mathfrak{b}$, define a ``layer'' $S_l$ as a set of all $v_i$s in the range $[(1+\zeta)^{q}(1+\epsilon)^{l}, (1+\zeta)^{q}(1+\epsilon)^{l+1})$.
Denote by $s_l$ the number of elements in $S_l$.
For any $l$ define a \emph{left boundary sub-layer} $S_{l,left}$ as a set of all $v_i$s in the range
$[(1+\zeta)^{q-1}(1+\epsilon)^{l}, (1+\zeta)^{q}(1+\epsilon)^{l})$ and $s_{l,left}$ to be its size.
For any $l$ define a \emph{right boundary sub-layer} $S_{l,right}$ as a set of all $v_i$s in the range
$[(1+\zeta)^{q}(1+\epsilon)^{l}, (1+\zeta)^{q+1}(1+\epsilon)^{l})$ and $s_{l,right}$ to be its size.
Let $\mathfrak{S}$ be the set of all element in boundary (left or right) sublayers. It is straightforward to see the total weight of the elements in $\mathfrak{S}$ is small, w.h.p.:
\begin{fact}\label{fact:boundary}
$P(\sum_{v_i \in \mathfrak{S}}v_i \ge {20\over Q} |V|) \le 0.1$.
\end{fact}
\begin{proof}
For a fixed $v_i$, let $j = Qx+y, 0\le y<Q$ be such that $(1+\zeta)^{j-1} < v_i \le (1+\zeta)^{j}$.
Then, $P(v_i \in \mathfrak{S}) = P(q-1\le y\le q) = {2\over Q}$.
Thus, by Markov inequality,
$P(\sum_{v_i \in \mathfrak{S}}v_i \ge {20\over Q}|V|) \le 0.1$.
\end{proof}

\noindent
\emph{Proof of Lemma \ref{lm: estimating V}}
We prove that Algorithm \ref{alg:approx zuzik finder} satisfies the requirement of the lemma.
Let $\mathcal{B}$ be the event that for all $j$, $\mathfrak{B_j}$ returns a $\mathfrak{c}$-cover of $VG_j$ (see Definition \ref{def: cover ddddff}).
By parameters of $\mathfrak{B_j}$ and by the union bound $P(\mathcal{B}) \ge 1-{\mathfrak{a}\over {\chi}'}$ for any fixed functions $G_j$.
Let $\mathcal{D}$ be the event that $\sum_{v_i \in \mathfrak{S}}v_i < {20\over Q} |V|$. By Fact \ref{fact:boundary}, we have $P(\mathcal{D}) \ge 0.9$.
In the remainder of this section we assume that $\mathcal{B, D}$ are true.
The key observation is that if $\mathcal{B}$ is true then any $v_i \notin \mathfrak{S}$ is not misclassified; i.e., if an approximation of $v_i$ is returned, then it will belong to the same layer as $v_i$.

\subsubsection*{\emph{\textbf{I. Upper Bound}}}
To prove the upper bound, we distinguish between \emph{large} and \emph{small} layers.
A layer $S_l$ is large if $\tilde{s}_l = s_l + s_{l,left} + s_{l,right} > {\chi}$, and small otherwise.
Consider a fixed $l$; if $S_l$ is a large layer, then Corollary \ref{cor:Z'} is applicable as follows.
Let $X_{i,j}$ be the indicator of the event that $G_j(i) = 1$, and let $v_{i_1}, \dots, v_{i_{\tilde{s}_l}}$
be the elements from $S_l \cup S_{l,left} \cup S_{l,right}$. Let $\alpha_{i,j}$ be the indicator random variable that the approximation of $v_i$ will be counted by$\mathcal{Y}_{l,j}$. Since $\mathcal{B}$ is true, no elements outside of $S_l \cup S_{l,left} \cup S_{l,right}$ can be counted. Thus, we can write
$$
\mathcal{Y}_{l,j} = \sum_{t=1}^{\tilde{s}_l} X_{i_t,j}\alpha_{i_t,j}
$$
and apply Corollary \ref{cor:Z'} with $n_{rest} = \tilde{s}_l$ and an appropriate enumeration of $X$s. Therefore, by Corollary \ref{cor:Z'}, w.p. at least $1-{1\over {\chi}'}$,
$$
\mathcal{Z}_l \le {\mathfrak{f}}_{\chi}(\tilde{s}_l) + 2.
$$
Consider the case that $\mathcal{Z}_l > 0$. Then, by definition of $\mathcal{Z}_l$, we have $\mathcal{Y}_{l,\mathcal{Z}_l} \le (1+3\epsilon){\chi}$, and thus by definition of ${\mathfrak{f}}_{\chi}$:
$$
(1+\epsilon)^{\mathcal{Z}_l}\mathcal{Y}_{l,Z_l} \le (1+\epsilon)^{{\mathfrak{f}}_{\chi}(\tilde{s}_l) + 2}(1+3\epsilon){\chi} \le (1+\epsilon)^6\tilde{s}_i.
$$
Also if $\mathcal{Z}_l = 0$ then $\mathcal{Y}_{l,0} \le \tilde{s}_l$, assuming $\mathcal{B}$. In this case we have $(1+\epsilon)^{\mathcal{Z}_i}Y_{i,Z_i} \le \tilde{s}_i.$
Thus, for any large layer, we have w.p. at least $1-{1\over {\chi}'}$:
$$
(1+\epsilon)^{\mathcal{Z}_l}\mathcal{Y}_{l,Z_l} \le (1+\epsilon)^6\tilde{s}_l.
$$
Consider the case when $S_l$ is small.
For the purposes of our analysis, we can add to $\mathcal{Y}_{l,j}$ arbitrary elements $v_{\tilde{s}_l+1}, \dots v_{{\chi} + 1} \notin S_l\cup  S_{l,left} \cup S_{l,right}$ and define
$\alpha_{i_t, j} \equiv 0$ for all $j$ and for all $t>\tilde{s}_l$. Thus, the above bounds will be valid.
Thus, we conclude that for every layer $S_l$ the approximation of its cardinality exceeds
$(1+\epsilon)^6\tilde{s}_l$ w.p. at most $1\over {\chi}'$. By union bound and by Fact \ref{fact:boundary}, w.p. at least $1-{\mathfrak{b}\over {\chi}'}$:
$$
(1+\zeta)^q\sum_{l \in [\mathfrak{b}] } (1+\epsilon)^{\mathcal{Z}_l+l}\mathcal{Y}_{l,\mathcal{Z}_l} \le (1+\zeta)^q\sum_{l\in [\mathfrak{b}]} (1+\epsilon)^{l+6}\tilde{s}_l \le
$$
$$
\sum_{i\in [n]} v_i(1+\epsilon)^7 + \sum_{v_i \in \mathfrak{S}} v_i(1+\epsilon)^7 \le (1+\epsilon)^7(1+20\epsilon)|V|.
$$

\subsubsection*{\emph{\textbf{II. Lower Bound}}}
Now let us prove the lower bound.
Assuming $\mathcal{B}$, the only elements from $S_l$ that cannot be counted by $\mathcal{Y}$ are those from $S_{l-1,right}$ and $S_{l+1,left}$.
Let $\hat{S}_l = S_l\setminus (S_{l-1,left}\cup S_{l+1,right})$ and let $\hat{s}_l = s_l - s_{l+1,left} - s_{l-1,right}$ to be its size.
We change a definition of a \emph{large} layer; $S_l$ is large if $\hat{s}_l > {\chi}$, and \emph{small} otherwise.
Consider an $\epsilon\over {\chi}'$-significant layer $\hat{S}_l$.

\emph{\textbf{II.1. large layers}}

First, let us assume that $\hat{S}_l$ is large.
Let $v_{i_1}, \dots, v_{i_{\hat{s}_l}}$ be elements from $\hat{S}_l$.
Let $X_{i,j} = \one_{G_j(i) = 1}$ and let $Y_{l, j} = \sum_{t = 1}^{\hat{s}_l} X_{i_t,j}$; i.e.,
$Y_{l,j}$ is the number of elements among $v_{i_1}, \dots, v_{i_{\hat{s}_l}}$ that has not been zeroed by $G_j$.
Consider an event $\mathcal{A}$ that ${{\chi}\over (1+\epsilon)^2} < Y_{l, {\mathfrak{f}}_{\chi}(\hat{s}_l)} \le {\chi}(1+3\epsilon)$.
By Fact \ref{lm: ddd upper bounds layers} we have
$$
P(\mathcal{A}) \ge 1-{1\over {\chi}'}.
$$
Let $\tilde{R} = \sum_{v_i\notin \hat{S}_l} G_{{\mathfrak{f}}_{\chi}(\hat{s}_l)}(i)v_i$ be the total weight of all elements that do not belong to $\hat{S}_l$ and contribute to $|VG_{{\mathfrak{f}}_{\chi}(\hat{s}_l)}|$. We have
$$
E(\tilde{R}) = {1\over (1+\epsilon)^{{\mathfrak{f}}_{\chi}(\hat{s}_l)}}\sum_{v_i\notin \hat{S}_l} v_i \le {|V|\over (1+\epsilon)^{{\mathfrak{f}}_{\chi}(\hat{s}_l)}}.
$$
Consider the event $\mathcal{C}$ that $\tilde{R} \le {{\chi}'|V|\over (1+\epsilon)^{{\mathfrak{f}}_{\chi}(\hat{s}_l)}}$. We have by Markov inequality that
$$
P(\mathcal{C}) \ge 1-{1\over {\chi}'}.
$$
Below we prove that all elements from $\hat{S}_l$
will belong to $\mathfrak{c}$-cover returned by $\mathfrak{Q}_{{\mathfrak{f}}_{\chi}(\hat{s}_l)}$.
Recall that for any $v_i \in \hat{S}_l$ we have $(1+\zeta)^q(1+\epsilon)^{l-1} < v_i \le (1+\zeta)^q(1+\epsilon)^{l}$. 
Thus, for every $v_i\in \hat{S}_l$ since $\hat{S}_l$ is $\epsilon\over {\chi}'$-significant, $\mathcal{C}$ is true and 
by definition of ${\mathfrak{f}}_{\chi}$:
$$
v_i \ge (1+\zeta)^q(1+\epsilon)^{l-1} \ge {\epsilon\over {\chi}'}{1\over (1+\epsilon)\hat{s}_l}|V| \ge {\epsilon\over {\chi}'^2}{1\over (1+\epsilon)\hat{s}_l}\tilde{R}(1+\epsilon)^{{\mathfrak{f}}_{\chi}(\hat{s}_l)}  \ge
$$
$$
{\epsilon\over 2\chi\chi'^2} \sum_{v_{i'}\notin \hat{S}_l} G_{{\mathfrak{f}}_{\chi}(\hat{s}_l)}(i)v_{i'}.
$$
Since $\mathcal{A}$ is true it follows that $Y_{l, {\mathfrak{f}}_{\chi}(\hat{s}_l)} \le {\chi}(1+3\epsilon)$. Thus,
$$
v_i \ge (1+\zeta)^q(1+\epsilon)^{l-1} \ge {Y_{l,{\mathfrak{f}}_{\chi}(\hat{s}_l)} \over 4{\chi}}(1+\zeta)^q(1+\epsilon)^{l-1} =
$$
$$
{1 \over 4{\chi}}\sum_{v_{i'}\in \hat{S}_l} G_{{\mathfrak{f}}_{\chi}(\hat{s}_l)}(i)(1+\zeta)^q(1+\epsilon)^{l-1} \ge
{1 \over 8{\chi}}\sum_{v_{i'}\in \hat{S}_l} G_{{\mathfrak{f}}_{\chi}(\hat{s}_l)}(i)v_{i'}.
$$
Thus, we conclude that
$$
v_i \ge {\epsilon\over 4\chi\chi'^2} |VG_{{\mathfrak{f}}_{\chi}(\hat{s}_l)}|.
$$
But this bound and $\mathcal{B}$ imply that all $v_i\in \hat{S}_l$ with $G_{{\mathfrak{f}}_{\chi}(\hat{s}_l)}(i) = 1$ will be
found by $\mathfrak{Q}_{{\mathfrak{f}}_{\chi}(\hat{s}_l)}$ and counted by $\mathcal{Y}_{l, {\mathfrak{f}}_{\chi}(\hat{s}_l)}$. Thus
\begin{equation} \label{eq:one y}
\mathcal{Y}_{l, {\mathfrak{f}}_{\chi}(\hat{s}_l)} \ge Y_{l,{\mathfrak{f}}_{\chi}(\hat{s}_l)} \ge {{\chi}\over (1+\epsilon)^2}.
\end{equation}
Let $\mathfrak{O}_l = S_{l,left}\cup S_{l-1,right}\cup S_{l+1,left}\cup S_{l,right} \subseteq \mathfrak{S}$.
Let $\mathfrak{o}_l$ be the number of elements in $\mathfrak{O}_l$.
Then since $\mathcal{D}$ is true, we have:
$$
(1+\zeta)^{q-1}(1+\epsilon)^l\hat{s}_l \ge \sum_{v_i\in \hat{S}_l} {v_i} \ge {\epsilon\over {\chi}'}|V| \ge {20\over \epsilon Q}|V| \ge {1\over \epsilon} \sum_{v_i \in \mathfrak{S}} v_i \ge
$$
$$
{1\over \epsilon} \sum_{v_i \in \mathfrak{O}_l} v_i \ge {1\over \epsilon} \mathfrak{o}_l (1+\zeta)^{q-1}(1+\epsilon)^{l-1}.
$$
Thus,
$$
\mathfrak{o}_l \le \epsilon(1+\epsilon) \hat{s}_l \le 2\epsilon \hat{s}_l.
$$
Consider $\check{Y} = \sum_{v_i \in \hat{S}_l \cup \mathfrak{O}_l} G_{{\mathfrak{f}}_{\chi}(\hat{s}_l)}(i)$.
Assuming $\mathcal{B}$, only elements from $\hat{S}_l\cup \mathfrak{O}_l$ can contribute to $\mathcal{Y}_{l, {\mathfrak{f}}_{\chi}(\hat{s}_l)}$, and thus
$\mathcal{Y}_{l, {\mathfrak{f}}_{\chi}(\hat{s}_l)} \le \check{Y}$.
Further, we have
$$
E(\check{Y}) = {\hat{s}_l + \mathfrak{o}_l \over (1+\epsilon)^{{\mathfrak{f}}_{\chi}(\hat{s}_l)} } \le {(1+2\epsilon)\hat{s}_l \over (1+\epsilon)^{{\mathfrak{f}}_{\chi}(\hat{s}_l)} } \le (1+2\epsilon){\chi}.
$$
Also, by pairwise independence of $G_{{\mathfrak{f}}_{\chi}(\hat{s}_l)}$, we have $Var(\check{Y}) \le E(\check{Y})$.
Thus, by Chebyshev inequality:
$$
P(\check{Y} > (1+3\epsilon){\chi}) = P(\check{Y} -E(\check{Y}) > (1+3\epsilon){\chi} - E(\check{Y})) \le
$$
$$
P(\check{Y} -E(\check{Y}) \ge \epsilon {\chi}) \le {Var(\check{Y}) \over \epsilon^2 {\chi}^2} \le {(1+2\epsilon) \over \epsilon^2 {\chi} } \le {1\over {\chi}'}.
$$
Therefore,
\begin{equation} \label{eq:two y}
P(\mathcal{Y}_{l, {\mathfrak{f}}_{\chi}(\hat{s}_l)} \ge (1+3\epsilon){\chi}) \le {1\over {\chi}'}.
\end{equation}
By $(\ref{eq:one y})$ and $(\ref{eq:two y})$, w.p. at least $1-{2\over {\chi}'}$ we have ${{\chi}\over (1+\epsilon)^2} \le \mathcal{Y}_{l, {\mathfrak{f}}_{\chi}(\hat{s}_l)} \le (1+3\epsilon){\chi}$,
in which case
 $\mathcal{Z}_{l} \ge {\mathfrak{f}}_{\chi}({\hat{s}_l}) > 0$ and thus by definitions of $\mathcal{Z}_l$ and ${\mathfrak{f}}_{\chi}$:
$$
(1+\epsilon)^{\mathcal{Z}_l}\mathcal{Y}_{l,\mathcal{Z}_l} \ge (1+\epsilon)^{{\mathfrak{f}}_{\chi}(\hat{s}_l)}{{\chi}\over (1+\epsilon)^2} \ge {\hat{s}_l\over (1+\epsilon)^2}.
$$

\emph{\textbf{II.2. small layers}}

Similarly, if $\hat{S}_l$ is small and $\mathcal{Z}_{l} > 0$ we have
$$
(1+\epsilon)^{\mathcal{Z}_l}\mathcal{Y}_{l,Z_l} \ge \mathcal{Y}_{l,Z_l} \ge {{\chi}\over (1+\epsilon)^2} \ge {\hat{s}_l\over (1+\epsilon)^2} .
$$
Otherwise if $Z_l = 0$ then, w.h.p. $Y_{l,0} \ge \hat{s}_l$.
Indeed, for every $v_i\in \hat{S}_l$ we have that
$$
v_i \ge {\epsilon\over {\chi}'}{1\over (1+\epsilon)\hat{s}_l}|V| \ge {\epsilon\over (1+\epsilon)\chi\chi'}|V|.
$$
Thus, $\mathfrak{Q}_0$ will return approximations of all elements from $\hat{S}_l$ w.p. at least $1-{1\over {\chi}'}$; and all approximations will be counted towards $\mathcal{Y}_{l,0}$.
Thus $(1+\epsilon)^{\mathcal{Z}_l}\mathcal{Y}_{l,Z_l} \ge \hat{s}_l$.

\emph{\textbf{II.3. putting it all together\ \ \ }}

By union bound, for all $l$ such that $\hat{S}_l$ is $\epsilon\over {\chi}'$-significant layers, w.p. at least $1-{\mathfrak{b}\over {\chi}'}$ we have
$$
(1+\epsilon)^{\mathcal{Z}_l}\mathcal{Y}_{l,Z_l} \ge {\hat{s}_l\over (1+\epsilon)^2}.
$$
Note that
$$
|V| = \sum_{v_i \in \mathfrak{S}} v_i + \sum_{l\in [\mathfrak{b}]}\sum_{v_i \in \hat{S}_l} v_i.
$$
Let $\mathcal{L}$ be the set of all $l$ such that $\hat{S}_l$ is $\epsilon\over {\chi}'$-significant.
Assuming $\mathcal{D}$, we have for sufficiently large $n$:
$$
\sum_{v_i \in \mathfrak{S}} v_i + \sum_{l\in \bar{\mathcal{L}}}\sum_{v_i \in \hat{S}_l} v_i \le
20{\epsilon^2\over\log{n}}|V| + \mathfrak{b}{\epsilon\over {\chi}'}|V| \le \epsilon |V|.
$$
We have obtained that w.p. at least $1-{\mathfrak{b}\over {\chi}'}$:
$$
(1+\zeta)^q\sum_{l \in [\mathfrak{b}]} (1+\epsilon)^{\mathcal{Z}_l+l}\mathcal{Y}_{l,\mathcal{Z}_l} \ge (1+\zeta)^q\sum_{l \in \mathcal{L}} (1+\epsilon)^{\mathcal{Z}_l+l}\mathcal{Y}_{l,\mathcal{Z}_l} \ge \sum_{l \in \mathcal{L}} (1+\zeta)^q(1+\epsilon)^l {\hat{s}_l\over (1+\epsilon)^2} \ge
$$
$$
\sum_{l \in \mathcal{L}}\sum_{v_i \in \hat{S}_l} { v_i\over (1+\epsilon)^2} \ge {(1-\epsilon)\over (1+\epsilon)^2}|V|.
$$

\subsubsection*{\emph{\textbf{III. Conclusion}}}
We have shown that, w.p. at least  $(0.9)(1-{\mathfrak{a}\over {\chi}'})(1-{2\mathfrak{b}\over {\chi}'}) > 2/3$, the output of Algorithm \ref{alg:approx zuzik finder}
is greater than or equal to ${(1-\epsilon)\over (1+\epsilon)^2}|V|$
and smaller than or equal to $(1+\epsilon)^7(1+20\epsilon)|V|$.
By replacing $\epsilon$ with an appropriate $\epsilon' = \Omega(\epsilon)$, we obtain an $\epsilon$-approximation of $|V|$.

\subsubsection*{\emph{\textbf{IV. Memory bounds}}}
We apply $\mathfrak{a}$ algorithms $\mathfrak{Q}$, thus the total memory required for these is
$\mathfrak{a}(\mu(n,m,\mathfrak{c},{1\over {\chi}'}))$. To generate pairwise-independent functions $H$, we need $O(\mathfrak{a}\log{n})$ memory bits. We also maintain $\mathfrak{a}\mathfrak{b}$ counters $\mathcal{Y}$.
In total, by Fact \ref{fact:zeta}, we need
$$
O\left({1\over \epsilon}\log(n)\mu(n,m,{\epsilon^7\over \log^3(nm)},{\epsilon \over \log(nm)}) + {1\over \epsilon^2}\log^2(nm)\right)
$$
memory bits.
\qed


\section{Proving Lemmas \ref{lm: epsilon for indep tensors} and \ref{lm: polylog for indep tensors}}\label{sec:indyk genralizing}

\noindent
\textbf{Lemma \ref{lm: epsilon for indep tensors}.}
\textit{There exists an algorithm $\mathfrak{B}_{k-1}$ that, given a data stream $D$ and an access to hash functions $H_1,\dots,H_{k-1}$, in one pass obtains an $\epsilon$-approximation of $|T_{k-1}(W(M_{Ind},H_1,\dots, H_{k-1}))|$  using memory $O({1\over \epsilon^2}\log{1\over \delta}\log{nm\over \epsilon\delta})$.
}
\begin{proof}

 For $j\in [n]$, define $C_{j}$ to be independent random variables with Cauchy distribution. For ${\mathbf{i}}\in [n]^{k-1},$ denote $\mathcal{H}({\mathbf{i}}) = \prod_{l=1}^{k-1} H_l({\mathbf{i}}_l)$. Define
$$
Z = \sum_{j=1}^n C_{j} \left(\sum_{{\mathbf{i}}\in [n]^{k-1}} m_{({\mathbf{i}},j)}\mathcal{H}({\mathbf{i}})\right)
$$
By the arguments from \cite{stable}, a median of $\Omega({1\over \epsilon^2}\log{1\over \delta})$ independent $Z$s is an $(\epsilon, \delta)$-approximation of
$$
\sum_{j=1}^n \left|\sum_{{\mathbf{i}}\in [n]^{k-1}} m_{({\mathbf{i}},j)}\mathcal{H}({\mathbf{i}})\right| = |T_{k-1}(W(M_{Ind},H_1,\dots, H_{k-1}))|.
$$
To construct $Z$ in a single pass over $D$, we follow the ideas from \cite{mi}. Define $k+1$ random variables $Joint, Margin_1, \dots, Margin_k$ to be initially equal to $0$ and to be updated as follows.
Upon receiving a $k$-tuple $({\mathbf{i}},j) \in [n]^k, {\mathbf{i}}\in [n]^{k-1}, j\in [n]$, we put $Joint_s = Joint_s + \mathcal{H}({\mathbf{i}})C_{j}$.
For $s<k$, we put $Margin_s = Margin_s + H_s({\mathbf{i}}_s)$. Finally we put $Margin_k = Margin_k + C_{j}$.
We have
$$
Joint = \sum_{j\in [n]} C_{j} \sum_{{\mathbf{i}}\in [n]^{k-1}} f_{({\mathbf{i}},j)}\mathcal{H}({\mathbf{i}}).
$$
Also, for $s<k$ we have
$$
Margin_s = \sum_{{\mathbf{i}}_s \in [n]} f_s({{\mathbf{i}}_s}) H_s({\mathbf{i}}_s).
$$
Finally
$$
Margin_k = \sum_{j\in [n]} f_{k}(j) C_j.
$$
Thus,
$$
\prod_{s=1}^k Margin_s = \left(\sum_{j} C_j f_k(j)\right)\left(\sum_{{\mathbf{i}}\in [n]^{k-1}} \mathcal{H}({\mathbf{i}})\prod_{s=1}^{k} f_{s}({\mathbf{i}}_s)\right) =
m^k \sum_{j} C_j \sum_{{\mathbf{i}}\in \in [n]^{k-1}} \mathcal{H}({\mathbf{i}}) P_{product}(({\mathbf{i}},j)).
$$
Thus,
$$
m^kW - \prod_{s=1}^k Margin_s = \sum_{j=1}^n C_{j} \left(\sum_{{\mathbf{i}}\in [n]^{k-1}} m_{({\mathbf{i}},j)}\mathcal{H}({\mathbf{i}})\right).
$$
What remains is to analyze the memory bounds. Recall that we don't count the memory of $\mathcal{H}$, which will be analyzed separately.
Thus, we need to bound a memory needed to compute $Z_s$.
To compute $Z$, our algorithm accesses $n$ random variables $C_{j}$ and computes a sketch that is a weighted sum of $C_{j}$.
Indyk shows in \cite{stable} (see Sections $3.2$ and $3.3$), that if the coefficients of $C_{j,s}$ are polynomially bounded integers, then
it is possible to maintain such a sum with sufficient precision using
$O(\log{nm\over \epsilon\delta})$ memory bits.
By Fact \ref{fct: bounded }, all entries of $T_s(W(M_{Ind}, \mathcal{H}))$ are polynomially bounded integers;
thus, we can repeat the arguments from \cite{stable} and the lemma follows.

\end{proof}

In the reminder of this paper, we assume that $\varpi = O(kn)$. A $\varpi$-truncated Cauchy variable $X$ is a modified Cauchy variable $Y$ such $X = -\varpi\one_{Y < -\varpi} + Y\one_{-\varpi\le Y \le \varpi} + \varpi\one_{Y > \varpi}$.

\begin{definition}\label{def:product sketch}
Let $C_{j,i}, j\in[t], i\in [n]$ be independent random variables where $C_{1,*}$ are Cauchy and
$C_{j,*}, j>1$ are $\varpi$-truncated Cauchy variables.
For every ${\mathbf{i}}\in [n]^t$ define $C({\mathbf{i}}) = \prod_{l=1}^t C_{l, {\mathbf{i}}_l}$.
A \emph{product sketch} of $t$-dimensional tensor $M$ (with entries $m_{\mathbf{i}}, \mathbf{i}\in [n]^t$) is
$$
\mathcal{C}(M) = \sum_{{\mathbf{i}}\in [n]^t} m_{\mathbf{i}} C({\mathbf{i}}).
$$
\end{definition}

\begin{lemma}\label{lm:product sketch}
It is possible to generate in one pass a product sketch of a tensor
$T_{s'}(W(M_{Ind}, H_1,\dots, H_s))$ for any $0\le s' \le s\le k$.
\end{lemma}
\begin{proof}
Generate $C_{j,i}, j\in [k-s'], i\in [n]$ random variables as in Definition \ref{def:product sketch}.
Consider $k+1$ variables $Joint,Margin_1,\dots,Margin_k$ initially zero and updated as follows:
compute
$$
Joint = Joint + \prod_{j\in [s]} H_{j}({\mathbf{i}}_{j})\prod_{j\in [k-s']} C_{j,{\mathbf{i}}_{s'+j}};
$$
and for $j\le s'$
$$Margin_j = Margin_j + H_{j}({\mathbf{i}}_{j});$$
and for $j>s$
$$Margin_j = Margin_j + C_{j-s',{\mathbf{i}}_{j}};$$
and for $s' < j\le s$
$$Margin_j = Margin_j + H_{j}({\mathbf{i}}_{j})C_{j-s',{\mathbf{i}}_{j}}.$$
At the end, we also compute $Product = \prod_{j=1}^k Margin_j$.
We consider the quantity $m^kJoint - Product$ written in the form $\sum_{{\mathbf{i}}\in [n]^{k-s'}} C(\mathbf{i}) Coef(\mathbf{i})$.
Our goal is to compare $Coef(\mathbf{i})$ with the entries of the tensor $T_{s'}(W(M_{Ind}, H_1,\dots, H_s))$.
Let ${\mathbf{i}}\in [n]^{k-s'}$ be fixed.
For $Joint$, a coefficient that corresponds to $C({\mathbf{i}})$  is equal to:
$$
\sum_{{\mathbf{j}}\in [n]^{s'}} f_{({\mathbf{j}},{\mathbf{i}})} (\prod_{l=1}^{s'} H_l({\mathbf{j}}_l))(\prod_{l=s'+1}^{s} H_l({\mathbf{i}}_{l-s'})). $$
For $Product = \prod Margin_j$, a coefficient that corresponds to $C({\mathbf{i}})$ is equal to:
$$
\sum_{{\mathbf{j}}\in [n]^{s'}} \left[(\prod_{l=1}^{s'} H_l({\mathbf{j}}_l))(\prod_{l=s'+1}^{s} H_l({\mathbf{i}}_{l-s'}))\prod_{l=1}^{s'}f_{l}({\mathbf{j}}_{l})\prod_{l=s'}^{k}f_{l}({\mathbf{i}}_{l-s'+1})\right] =
$$
$$
m^k\sum_{{\mathbf{i}}\in [n]^{s'}} P_{product}(({\mathbf{i}},{\mathbf{j}}))(\prod_{l=1}^{s'} H_l({\mathbf{j}}_l))(\prod_{l=s'+1}^{s} H_l({\mathbf{i}}_{l-s'})).
$$
Thus, the coefficient of $C({\mathbf{i}})$ in
$
m^kJoint-Product
$
is
$$
\sum_{{\mathbf{j}}\in [n]^{s'}} m_{({\mathbf{j}},{\mathbf{i}})}(\prod_{l=1}^{s'} H_l({\mathbf{j}}_l))(\prod_{l=s'+1}^{s} H_l({\mathbf{i}}_{l-s'})).
$$
On the other hand, consider $T_{s'}(W(M_{Ind}, H_1,\dots, H_s))$.
The coefficient of $W(M_{Ind}, H_1,\dots, H_s)$ is $m'_{\mathbf{i}} = m_{\mathbf{i}} \prod_{l=1}^s H_l({\mathbf{i}}_l)$.
Thus, the coefficient of $T_{s'}(W(M_{Ind}, H_1,\dots, H_s))$ is for ${\mathbf{i}}\in [n]^{k-s'}$:
$$
\sum_{{\mathbf{j}}\in [n]^{s'}} m'_{({\mathbf{j}},{\mathbf{i}})} = \sum_{{\mathbf{j}}\in [n]^{s'}} m_{({\mathbf{j}},{\mathbf{i}})}(\prod_{l=1}^{s'} H_l({\mathbf{j}}_l))(\prod_{l=s'+1}^{s} H_l({\mathbf{i}}_{l-s'})).
$$
Thus $m^kJoint-Product$ is the product sketch for $T_{s'}(W(M_{Ind}, H_1,\dots, H_s))$.
It is important to note that the procedure above works for $s' = 0$ as well.
\end{proof}

\begin{fact}\label{fct: independ cauchy}
Let $C_1,\dots,C_n$ be independent Cauchy variables and let $\alpha_1,\dots,\alpha_n$ be arbitrary random variables independent of $C_1,\dots, C_n$.
Then
$$
P(|\sum_{i} C_i\alpha_i| \le {|\alpha|\over t}) \le {1\over t}.
$$
\end{fact}
\begin{proof}
By stability, we have $\sum_{i} C_i\alpha_i \sim C|\alpha|$, where $C$ is a Cauchy variable. Thus,
$$
P(|C||\alpha| \le {1\over t}|\alpha|) \le {1\over \pi}\int_{-{1\over t}}^{{1\over t}}{1\over 1+x^2} \le {1\over t}.
$$
\end{proof}

\begin{fact}\label{fct: non-negative real}
Let $\{a_1, \dots, a_{\mathfrak{n}}\}$ be non-negative real numbers and let $X_i, i\in [\mathfrak{n}]$ be  non-negative random variables such that $P(X_i \le a_i) \le {1\over \mathfrak{q}}$.
Let $X = \sum_{i\in [\mathfrak{n}]}X_i$.
Then
$$
P(X \le {1\over 2} \sum_{i}a_i) \le {2\over \mathfrak{q}}.
$$
\end{fact}
\begin{proof}
Let $Y_i = a_i$ if $X_i \ge a_i$, and $Y_i = 0$ otherwise.
Then
$$
E(Y_i) = a_iP(X_i \ge a_i) \ge a_i(1-{1\over \mathfrak{q}}).
$$
Let $Z_i = a_i - Y_i$. Then $Z_i \ge 0$ and $E(Z_i) \le {a_i \over \mathfrak{q}}$.
Let $Y= \sum_{i} Y_i, Z = \sum_{i} Z_i$. Then by Markov inequality,
$$
P(Z\ge {\mathfrak{q}'\over \mathfrak{q}}\sum_{i} a_i) \le {1\over \mathfrak{q}'}.
$$
Thus
$$
P(\sum_{i} a_i - Y \ge {\mathfrak{q}'\over \mathfrak{q}}\sum_{i} a_i) \le {1\over \mathfrak{q}'}.
$$
Thus
$$
P(X \le (1-{\mathfrak{q}'\over \mathfrak{q}})\sum_{i} a_i) \le P(Y \le (1-{\mathfrak{q}'\over \mathfrak{q}})\sum_{i} a_i) \le {1\over \mathfrak{q}'}.
$$
Putting $\mathfrak{q}' = {\mathfrak{q}\over 2}$, we obtain
$$
P(X \le {1\over 2}\sum_{i} a_i) \le {2\over \mathfrak{q}}.
$$
\end{proof}

\begin{lemma}\label{lm:cuachy lower bound}
Let $Y = \sum_{\mathbf{i}\in [n]^k} \prod_{j=1}^k C_{j, \mathbf{i}_j} m_{\mathbf{i}}$
where all $C$ are Cauchy.
For any $M$ with entries $m_{\mathbf{i}}$ and for $\mathfrak{q} > 3^k$ we have
$$
P(|Y| \le {|M| \over {(2\mathfrak{q})}^{k}}) \le {3^k \over \mathfrak{q}}.
$$
\end{lemma}
\begin{proof}
We prove the claim by induction on $k$.
For $k=1$ we have by  Fact \ref{fct: independ cauchy}:
$$
P(|\sum_{l\in [n]} C_{1,l} m_l| \le {|M| \over {2\mathfrak{q}}}) \le {3\over {\mathfrak{q}}}.
$$
Consider $k>1$.
For simplicity of presentation, put $C_l = C_{1,l}$ and
$$
Y_l = \sum_{\mathbf{i}\in [n]^{k-1}} \prod_{j=2}^k C_{j, (l,\mathbf{i}_{j-1})} m_{(l,\mathbf{i})}.
$$
Then
$$
Y = \sum_{l\in [n]} C_lY_l.
$$
We have, by stability of $C_l$s that, $\sum_{l\in [n]} C_lY_l \sim C'\sum_{l\in [n]}|Y_l|$
where $C'$ is Cauchy distributed. Thus
$$
P(|\sum_{l\in [n]} C_lY_l| \ge {|M|\over {(2\mathfrak{q})}^{k}}) = P(|C'|\sum_{l\in [n]} |Y_l| \ge {|M|\over {(2\mathfrak{q})}^{k}}) \ge
$$
$$
P(|C'|\sum_{l\in [n]} |Y_l| \ge{|M|\over  {(2\mathfrak{q})}^{k}}, \sum |Y_l| \ge {|M|\over 2^k{\mathfrak{q}}^{k-1}}) \ge
$$
$$
P(|C'|\sum_{l\in [n]} |Y_l| \ge{\sum|Y_l|\over {\mathfrak{q}}}, \sum |Y_l| \ge {|M|\over 2^k{\mathfrak{q}}^{k-1}}).
$$
We have, by Fact \ref{fct: independ cauchy}:
$$
P(|C'|\sum |Y_l| \le {\sum|Y_l|\over {\mathfrak{q}}}) \le {1\over \mathfrak{q}}.
$$
Denote by $M_l$ the $l$-th hyperplane of $M$.
By induction for each $l$:
$$
P(|Y_l| \le {|M_l|\over {(2\mathfrak{q})}^{k-1}}) \le {3^{k-1}\over \mathfrak{q}}.
$$
Thus, by Fact \ref{fct: non-negative real}:
$$
P(\sum |Y_l| \le {1\over 2}{|M|\over {(2\mathfrak{q})}^{2(k-1)}})  \le {2*3^{k-1}\over \mathfrak{q}}.
$$
By union bound, and since ${1\over \mathfrak{q}} + {2*3^{k-1}\over \mathfrak{q}} \le {3^k\over \mathfrak{q}}$,
the claim is correct.
\end{proof}

\begin{corollary}\label{cor:lower bound cauchy k-wise}
Let $Y = \sum_{\mathbf{i}\in [n]^k} \prod_{j=1}^k C_{j, \mathbf{i}_j} m_{\mathbf{i}}$
where all $C_{j,*}, j>1$ are $\varpi$-truncated Cauchy and all $C_{1,*}$ are Cauchy.
For any $M$ with entries $m_{\mathbf{i}}$ we have
$$
P(|Y| \le {|M| \over {200^k3^{k^2}}}) \le {1\over 50}.
$$
\end{corollary}
\begin{proof}
Consider an event that no $C$s is equal to $\varpi$.
Repeating the arguments from \cite{mi}, the probability that this event does not occur is bounded by
$$
{2kn\over \pi}\int_{-\infty}^{-\varpi} {1\over 1+x^2} \le {2kn\over \varpi\pi} \le {1\over 100}.
$$
Thus, and by Lemma \ref{lm:cuachy lower bound}:
$$
P(|Y| \le {|M| \over {(2\mathfrak{q})}^{k}}) \le {1\over 100} + {1\over 100}
$$
for $\mathfrak{q} = {100*3^k}$.
\end{proof}

\begin{lemma}\label{lm: product sketch is logk}
Let $M$ be a $s$-dimensional tensor for $s\le k$ and let $Y$
be a product sketch of $M$. I.e.,
$$
Y = \sum_{\mathbf{i}\in [n]^k} \prod_{j=1}^k C_{j, \mathbf{i}_j} m_{\mathbf{i}},
$$
where for all $j\in [k], i\in [n]$ the random variables $C_{j,i}$ are independent and
$C_{1,*}$ are Cauchy and $C_{j,*}, j>1$ are truncated Cauchy.
Then  $|Y|$ is a $\log^k{n}$-approximation of $|M|$ w.p. at least $0.07$.
\end{lemma}
\begin{proof}
We consider $s=k$; the same arguments can be repeated for any $s<k$.
Consider $Y_l = |\sum_{\mathbf{i} \in [n]^{k-1}} \prod_{j=2}^k C_{j, \mathbf{i}_j} m_{(l,\mathbf{i})} |,$ and let $Y' = \sum_{l\in [n]} Y_l$.
Indyk \cite{stable} shows that for any $C$ with $\varpi$-truncated Cauchy distribution, it is true that $E(|C|) \le \log{(\varpi^2 + 1)}/\pi + O(1)$. Thus, and by the independency of all $C$s, we have:
$$
E(Y_l) = E(|\sum_{\mathbf{i} \in [n]^{k-1}} \prod_{j=2}^k C_{j, \mathbf{i}_{j-1}} m_{(l,\mathbf{i})} |) \le \sum_{\mathbf{i} \in [n]^{k-1}} E(|\prod_{j=2}^k C_{j, \mathbf{i}_{j-1}}|) |m_{(l,\mathbf{i})}| =
$$
$$
\sum_{\mathbf{i} \in [n]^{k-1}} \prod_{j=2}^k E(|C_{j, \mathbf{i}_{j-1}}|) |m_{(l,\mathbf{i})}| \le 3\log^{k-1}{n}\sum_{\mathbf{i} \in [n]^{k-1}} |m_{(l,\mathbf{i})}| .
$$
Thus, by Markov inequality:
$$
P( |Y'| > 300\log^{k-1}{n}|M|)\le {1\over 100}.
$$
Since $|Y| \le |Y'|$ the upper bound follows.
The lower bound follows from Corollary \ref{cor:lower bound cauchy k-wise} and since for large enough $n, \log{n} > 200*3^k$.
\end{proof}

\noindent
\textbf{Lemma \ref{lm: polylog for indep tensors}.}
\textit{There exists an algorithm $A_{s_1, s_2}$ (for any $0\le s_2 \le s_1\le k$) that, given a data stream $D$ and
an access to hash functions $H_1,\dots,H_{s_1}$, in one pass obtains a $\log^k{n}$-approximation of $|T_{s_2}(W(M_{Ind},H_1,\dots,H_{s_1}))|$ using memory $O(\log{(nm)}\log{1\over \delta})$.
}
\begin{proof}
By Lemma \ref{lm:product sketch}, it is possible to construct a product sketch for
$|T_{s_2}(W(M_{Ind},H_1,\dots,H_{s_1}))|$ in one pass.
Also, by Lemma \ref{lm: product sketch is logk} the constructed product sketch is a $\log^k{n}$-approximation of $|T_{s_2}(W(M_{Ind},H_1,\dots,H_{s_1}))|$ w.p. $\Omega(1)$.
Thus, taking a median $O(\log{1\over \delta})$ of independent product sketches results in a $(\log^k{n}, \delta)$-approximation.
It remains to analyze the memory bounds.
Repeating the arguments from \cite{mi}, each product sketch can be constructed with sufficient precision using $O(k\log{nm})$ memory bits.
Also, the perfectly random variables can be replaced by pseudorandom variables and
using the ``sorting'' argument from \cite{mi} (Section $3.2$). We repeat the arguments of Indyk and McGregor $k$ times (instead of two as in \cite{mi}).
\end{proof}

\section*{Acknowledgments}

We thank Piotr Indyk for a helpful discussion.


\end{document}